\documentclass[letterpaper, 10 pt, conference]{ieeeconf}%

\IEEEoverridecommandlockouts
\usepackage{cite}
\usepackage{amsmath,amssymb,amsfonts}
\usepackage{algorithmic}
\usepackage{graphicx}
\usepackage{textcomp}
\usepackage{xcolor}
\usepackage{cuted}
\usepackage{lipsum}
\usepackage{mathtools}
\usepackage{stfloats}
\usepackage{bbm}
\usepackage{multirow}
\newcommand{\R}{\mathbb{R}}
\newcommand{\N}{\mathbb{N}}
\newcommand{\T}{\top}
\newcommand{\I}{\mathbf{I}}
\newcommand{\0}{\mathbf{0}}
\newcommand{\E}{\mathcal{E}}

\newcommand{\X}{\mathcal{X}}
\newcommand{\Y}{\mathcal{Y}}
\newcommand{\x}{{\scriptstyle \X}}
\newcommand{\y}{{\scriptstyle \Y}}
\newcommand{\e}{\mathrm{e}}

\newcommand{\diag}{\text{diag}}

\newcommand{\mb}[1]{\mathbf{#1}}

\newcommand{\bm}[1]{\begin{bmatrix}#1\end{bmatrix}}
\def\BibTeX{{\rm B\kern-.05em{\sc i\kern-.025em b}\kern-.08em
    T\kern-.1667em\lower.7ex\hbox{E}\kern-.125emX}}

\usepackage{subcaption}

\newtheorem{assumption}{\textbf{Assumption}}

\newtheorem{corollary}{\textbf{Corollary}}
\newtheorem{lemma}{\textbf{Lemma}}
\newtheorem{proposition}{\textbf{Proposition}}
\newtheorem{remark}{\textbf{Remark}}
\newtheorem{definition}{\textbf{Definition}}

\usepackage{hyperref}
    
\begin{document}

\title{
{\LARGE \textbf{
{Model-Based and Data-Driven Hierarchical Control and Topology Co-Design for Robust Networked Systems}
}}}



\author{Shirantha Welikala, Zihao Song, Hai Lin, and Panos J. Antsaklis
\thanks{Shirantha Welikala is with the Department of Electrical and Computer Engineering, School of Engineering and Science, Stevens Institute of Technology, Hoboken, NJ 07030, \texttt{{\small \{swelikal\}@stevens.edu}}. Zihao Song, Hai Lin and Panos J. Antsaklis are with the Department of Electrical Engineering, University of Notre Dame, Notre Dame, IN, USA. e-mail: \{{\tt zsong2, hlin1, pantsakl}\}{\tt @nd.edu}.}}

\maketitle


\begin{abstract}
In this paper, we consider a class of networked systems comprising an interconnected set of linear subsystems, disturbance inputs, and performance outputs. Using dissipativity theory, we first propose a model-based hierarchical control design strategy to ensure the closed-loop networked system is dissipative from its disturbance inputs to performance outputs. This involves designing local controllers for each subsystem to enforce local dissipativity guarantees, which are then exploited to co-design distributed global controllers and the interconnection topology to enforce global dissipativity guarantees while optimizing interconnection topology costs. The overall design process requires only solving a sequence of linear matrix inequality (LMI) problems, thereby retaining compositionality and decentralizability while avoiding non-convex, iterative design processes that are inefficient and centralized. This model-based hierarchical control design strategy assumes the knowledge of the subsystem dynamics, which may not hold in many real-world networked systems. Motivated by this, we also propose a data-driven hierarchical control design strategy that assumes only the availability of rich input-state-output trajectory data from the subsystems. The proposed data-driven design process assumes that the unknown disturbances affecting the subsystem dynamics are bounded by a quadratic matrix inequality (relaxing conventional bounds) and accounts for this by using the matrix S-lemma. Finally, the effectiveness of the proposed model-based and data-driven hierarchical control designs is illustrated for a networked system representing a DC microgrid, with the aim of enforcing robust (dissipative) voltage regulation and current sharing.  
\end{abstract}


\section{Introduction}
Large-scale networked systems arise in many engineering domains, including microgrids \cite{Najafirad2025P1,Najafirad2025P2}, electronic circuits \cite{Jafarian2019,Jeltsema2005}, robotic networks \cite{song2024port}, and vehicular platoons \cite{song2025distributed,Karafyllis2021,Antonelli2013}. These systems consist of multiple subsystems that interact via communication (cyber) and/or physical interconnection topologies and are often subject to external disturbances. Efficient and scalable design of such ``networked'' systems while jointly optimizing (enforcing) interconnection and control-theoretic costs (constraints) remains a fundamental challenge. For example, how to design interconnection topologies and controllers in a networked system to optimize total interconnection cost and $L_2$-gain (a disturbance robustness metric), while ensuring stability and dissipativity of the closed-loop networked system?

From a control-theoretic perspective, early efforts mainly focused on centralized control \cite{zhang2001stability,silva2008optimal}, in which a global controller is designed using full network information. While such formulations can provide strong performance guarantees, they generally become difficult to implement in large-scale systems due to limitations in scalability, communication, and information sharing. This has motivated a significant shift toward decentralized \cite{rotkowitz2005characterization} and distributed \cite{olfati2007consensus} control strategies, in which each subsystem is controlled using local information and, possibly, information exchanged with neighboring subsystems. 

In the context of distributed control, graph-theoretic methods, such as consensus-based methods \cite{bai2024distributed}, have been widely studied for network coordination. Moreover, optimization-based methods, including distributed model predictive control \cite{wei2024robust,bai2025distributed} and primal-dual methods \cite{maggiar2025consensus}, have been developed to systematically handle constraints and performance objectives. These approaches enable the formulation of control problems as distributed optimization problems \cite{neal2011distributed}, for which iterative algorithms are frequently used in real-time implementations. Robust \cite{boyd1994linear,zhou1996robust} and adaptive \cite{ioannou1996robust} control techniques have also been extensively studied to address uncertainties, disturbances, and time-varying dynamics in networked systems. However, these methods typically treat the interconnection topology as fixed or impose structural assumptions on it, and they do not directly provide a compositional mechanism for jointly designing local and global controllers and the interconnection topology under global dissipativity constraints. As a result, existing approaches may lead to conservative designs or require centralized, iterative, and potentially non-convex design procedures that are inefficient and not scalable.

In recent years, dissipativity theory has emerged as a powerful framework for analyzing and synthesizing control strategies for networked systems due to its compositional and input-output nature. In particular, dissipativity enables local subsystem-level guarantees to be systematically propagated to global network-level guarantees, making it well suited for scalable and distributed design. However, existing dissipativity-based control design methods often rely on centralized formulations \cite{bahmani2020lmi} or assume fixed interconnection topologies \cite{sun2023dissipativity,nguyen2024improvement}, limiting their scalability and adaptability. Hence, there is a need for hierarchical control design frameworks that can systematically synthesize local subsystem controllers together with distributed global coordination mechanisms under dissipativity constraints. Furthermore, despite the central role of interconnection topology in network performance and implementation cost, the joint design of controllers and interconnection topologies under dissipativity guarantees remains largely unexplored within a unified, tractable framework.

Another important limitation is that most of the aforementioned works require precise knowledge of subsystem dynamics, which may not be available in practical networked systems. Although learning-based and data-driven methods have been developed to reduce reliance on explicit models \cite{kosaraju2021reinforcement,gama2022distributed}, embedding stability, robustness, dissipativity, and topology-design constraints directly into the learning or design process remains challenging. This difficulty becomes more pronounced when subsystem data are affected by unknown disturbances, since the resulting controllers must be robust to all dynamics that are consistent with the measured data. Therefore, extending dissipativity-based hierarchical control and topology co-design to data-driven settings, especially in the presence of disturbances and without explicit subsystem models, is both practically important and theoretically challenging.

Building on the foundational results of Willems' fundamental lemma \cite{willems2005note}, data-driven control approaches enable controller synthesis directly from measured input-output trajectories, bypassing explicit system identification while still providing theoretical guarantees for stability and performance. Recent advances have developed rigorous analytical tools for characterizing system properties directly from data \cite{de2019formulas} and for synthesizing controllers using data-dependent convex conditions \cite{berberich2020data}. These developments make direct data-driven control particularly attractive for networked systems whose subsystem dynamics are difficult to model accurately, but for which rich trajectory data can be collected.

Several analytical tools have played an important role in making such data-driven synthesis problems tractable. For example, Finsler's lemma \cite{van2021matrix} has been used to derive data-dependent conditions for controller synthesis under both noise-free and noisy data. Similarly, the matrix S-lemma \cite{van2020noisy} enables the direct design of feedback controllers from noisy data by converting robust data-consistency requirements into tractable matrix inequalities. These tools provide a systematic way to handle quadratic inequalities arising from measured trajectories and disturbance/noise bounds, often leading to less conservative formulations than classical robust-control approaches. In this direction, quadratic-constraint-based methods \cite{coulson2019data,ma2023data} encode system properties through matrix inequalities derived from data, while dissipativity-oriented dualization methods \cite{kristovic2024state,kristovic2024output} provide systematic procedures for designing state-feedback and output-feedback controllers that guarantee dissipativity directly from data.

Despite these advances, most existing data-driven control methods focus primarily on controller synthesis for a single system or for systems with fixed interconnection structures. They do not directly address the hierarchical setting in which local subsystem controllers, distributed global controllers, and the interconnection topology must be designed together. This gap is important because, in large-scale networked systems, the topology determines not only implementation cost and communication burden but also the achievable robustness and dissipativity properties of the closed-loop system. Therefore, a natural and necessary extension is to develop data-driven hierarchical co-design methods that synthesize both controllers and interconnection topologies for networked systems, while retaining convexity, robustness to data disturbances, and compositional dissipativity guarantees.

A closely related line of work focuses directly on dissipativity-based interconnection analysis and topology synthesis. The compositional dissipativity frameworks in \cite{Arcak2016,Arcak2022} provide centralized LMI conditions for certifying stability and dissipativity of nonlinear networked systems from subsystem-level dissipativity properties, while \cite{Agarwal2021,WelikalaP32022} developed decentralized and compositional analysis procedures for linear networked systems with known subsystem models. However, direct synthesis of interconnection topologies has received comparatively less attention. Early results such as \cite{Xia2014,Xia2018} provide nonlinear-inequality-based synthesis conditions for networks with one or two subsystems, whereas \cite{Ebihara2017} addresses topology synthesis for positive linear systems with known dynamics. Related symmetry-based methods reduce analysis complexity or tune a restricted interconnection structure rather than treating a general topology as a design variable \cite{Rufino2018,Ghanbari2016,Goodwine2013}. Thus, although these works establish important foundations, existing topology-synthesis results either require explicit subsystem models, are limited to special system classes or small networks, or do not jointly synthesize controllers and topology under scalable dissipativity certificates.

A preliminary version of this work \cite{WelikalaP52022} partially addressed this gap by developing centralized, and later decentralized \cite{Welikala2022Ax3}, LMI techniques for dissipativity-based analysis and interconnection topology synthesis using pre-certified subsystem dissipativity properties. That formulation is useful when subsystems are already equipped with local controllers and the main design freedom is the interconnection matrix. However, it does not address how such subsystem dissipativity certificates can be enforced via feedback control, how distributed global controllers should be synthesized in conjunction with the topology, or how the design can be performed when subsystem models are unavailable. This paper addresses these limitations for networked linear systems with disturbance inputs and performance outputs. In particular, we propose a hierarchical co-design framework in which local controllers enforce favorable subsystem dissipativity properties, distributed global controllers and interconnection topologies are then co-designed to guarantee network-level dissipativity, and the same design logic is further extended to a direct data-driven setting using trajectory data, quadratic disturbance descriptions, and the matrix S-lemma. The proposed model-based and data-driven methods are illustrated on a DC microgrid case study, with the objective of achieving robust (dissipative) voltage regulation and current sharing.

\paragraph*{\textbf{Contributions}}
Our main contributions are as follows:
\begin{enumerate}
\item We formulate a hierarchical dissipativity-based co-design problem for networked linear systems, in which local controllers, distributed global controllers, and interconnection topologies are jointly synthesized.
\item We develop a model-based LMI framework that enforces subsystem-level dissipativity via local feedback design and leverages the resulting certificates to co-design distributed controllers and topologies with network-level dissipativity guarantees.
\item We derive a direct data-driven counterpart using only subsystem input-state-output trajectories affected by unknown but quadratically bounded disturbances, leading to robust LMI conditions via the matrix S-lemma.
\item We demonstrate the proposed model-based and data-driven methods on a DC microgrid example for robust voltage regulation, current sharing, and communication-topology design.
\end{enumerate}

\paragraph*{\textbf{Organization}} 
This paper is organized as follows. Section \ref{Preliminaries} introduces the notation, matrix-theoretic tools, dissipativity notions, LMI conditions for dissipativity of linear time-invariant systems, and the baseline networked-system interconnection design result used later in the paper. Section \ref{Sec:Model-Based} presents the model-based hierarchical co-design framework, including local controller synthesis, distributed global controller/topology co-design, and a refined local design that improves compatibility with the global stage. Section \ref{Sec:Data-Driven} develops the corresponding data-driven framework using subsystem trajectory data and quadratic disturbance bounds. Section \ref{Sec:Results} applies the proposed methods to a DC microgrid case study, and Section \ref{Sec:Conclusion} concludes the paper.

\paragraph*{\textbf{Notations}}
$\mathbb{R}$ and $\mathbb{N}$ denote the sets of real and natural numbers, respectively, and $\mathbb{N}^0 \triangleq \mathbb{N}\cup\{0\}$. 
For any $N\in\mathbb{N}$, we define $\mathbb{N}_N\triangleq\{1,2,\ldots,N\}$ and $\mathbb{N}_N^0\triangleq\{0,1,2,\ldots,N\}$. 
An $n\times m$ block matrix $A$ is denoted by $A=[A_{ij}]_{i\in\mathbb{N}_n,j\in\mathbb{N}_m}$. 
The notation $[A_{ij}]_{j\in\mathbb{N}_m}$ denotes a block row matrix, and $\diag([A_{ii}]_{i\in\mathbb{N}_n})$ denotes a block diagonal matrix. 
The matrices $\0$ and $\I$ denote zero and identity matrices, respectively, with dimensions clear from context. 
For a symmetric matrix $A\in\mathbb{R}^{n\times n}$, $A>0$ and $A\geq0$ denote positive definiteness and positive semidefiniteness, respectively. Throughout the paper, matrix inequalities are used only for symmetric matrices.
The symbol $\star$ denotes entries or blocks that are implied by symmetry. 
For a matrix $A$, we define $\mathcal{H}(A)\triangleq A+A^\T$. 
The notation $\mb{1}_{\{\cdot\}}$ denotes the indicator function, and $\e_{ij}\triangleq\mb{1}_{\{i=j\}}$. 
The norm $\|\cdot\|_1$ denotes the entrywise $\ell_1$ norm when applied to a matrix. 
The acronyms RHS (LHS) denote the right(left)-hand side of an equation.

\section{Preliminaries}\label{Preliminaries}

\subsection{Some Useful Lemmas} 
Here we recall several useful lemmas regarding matrix definiteness.

\begin{lemma} [Schur's complement \cite{Bernstein2009}] \label{Lm:SchursComplement}
For some matrices $\Theta, \Phi, \Gamma$ (with appropriate dimensions),  
$$
\{\Gamma-\Phi^\T \Theta^{-1} \Phi \geq 0,\ \Theta > 0\} \iff   
\bm{\Theta & \Phi \\ \Phi^\T & \Gamma} \geq 0.
$$
\end{lemma}

\begin{lemma} [Congruence principle \cite{Bernstein2009}] \label{Lm:Congruence}
A matrix $W \geq 0$ if and only if $P W P^\T \geq 0$ where $P$ is nonsingular.
\end{lemma}

\begin{corollary}\label{Co:SchurCongruence}
For some block matrices $\Theta \triangleq [\Theta_{ij}]_{i,j\in\N_2}$, $\Phi \triangleq \bm{\Phi_1 \\ \Phi_2^\T}$, and $\Gamma$ (with appropriate dimensions)    
\begin{align*}
\bm{\Theta_{11} & \Phi_{1} & \Theta_{12}\\ \Phi_1^\T & \Gamma & \Phi_2 \\ \Theta_{12}^\T & \Phi_2^\T & \Theta_{22}} \geq 0 
\iff \bm{\Theta_{11} & \Theta_{12} & \Phi_{1}  \\ \Theta_{12}^\T  & \Theta_{22} & \Phi_2^\T \\ \Phi_1^\T & \Phi_2 & \Gamma} \geq 0  \\
\iff \{\Gamma - \Phi^\T \Theta^{-1} \Phi \geq 0,\ \Theta > 0\}.
\end{align*}
\end{corollary}
\begin{proof}
First, by applying Lm. \ref{Lm:Congruence}, we can interchange block rows/columns $2$ and $3$ via selecting $P$ to be the corresponding permutation matrix. Then, the proof is complete by applying Lm. \ref{Lm:SchursComplement}. 
\end{proof}

\begin{lemma} [Block-element-wise form \cite{WelikalaJ22022}] \label{Lm:BEW}
Consider an $m\times m$ block matrix $\Psi \triangleq [\Psi^{kl}]_{k,l\in\N_m}$, where each block $\Psi^{kl}=[\Psi^{kl}_{ij}]_{i,j\in\N_n}$ is itself an $n \times n$ block matrix with compatible dimensions. The block-element-wise (BEW) form of $\Psi$ is defined as $\mathrm{BEW}(\Psi) \triangleq [[\Psi^{kl}_{ij}]_{k,l\in\N_m}]_{i,j\in\N_n}$, and it satisfies: 
$\Psi \geq 0  \iff \mathrm{BEW}(\Psi) \geq 0.$
\end{lemma}

\begin{lemma}\label{Lm:MatrixSLemma} [Matrix S-Lemma \cite{van2020noisy}]
Let $M, N$ be symmetric matrices, 
$$\mathcal{S}_N \triangleq \{ Z\in\R^{n \times k} : \bm{\I \\ Z}^\T N \bm{\I \\ Z} \geq 0\},$$
and $\mathcal{S}_N^+$ be defined similarly to $\mathcal{S}_N$ except with a strict inequality. Assume that there exists some $\bar{Z} \in \mathcal{S}_N^+$. Then, if $\mathcal{S}_N$ is bounded, the implication (for any matrix $Z$)
$$
\bm{\I \\ Z}^\T N \bm{\I \\ Z} \geq 0 \implies \bm{\I \\ Z}^\T M \bm{\I \\ Z} > 0
$$
holds if and only if there exists $\lambda \geq 0$ such that $M - \lambda N > 0$. If $\mathcal{S}_N$ is not bounded, 
the implication (for any matrix $Z$)
$$
\bm{\I \\ Z}^\T N \bm{\I \\ Z} \geq 0 \implies \bm{\I \\ Z}^\T M \bm{\I \\ Z} \geq 0
$$
holds if and only if there exists $\lambda \geq 0$ such that $M - \lambda N \geq 0$.
\end{lemma}



\subsection{Dissipativity}\label{SubSec:dissipativity}

Consider the discrete-time dynamical system  
\begin{equation}\label{Eq:GeneralSystem}
\begin{aligned}
    x(t+1) &= f(x(t),u(t)),\\
    y(t) &= h(x(t),u(t)),
    \end{aligned}
\end{equation}
where the state $x(t)\in\R^{n_x}$, input $u(t)\in \R^{n_u}$, output $y(t)\in\R^{n_y}$ and $t\in\N^0$. Suppose under $u(t)=\0$, $x(t)=\0$ is an equilibrium point, i.e., $\0 = f(\0,\0)$, and $h(\0,\0) = \0$.

\begin{definition}\label{Def:Dissipativity}
The system \eqref{Eq:GeneralSystem} is \emph{dissipative} (from input $u$ to output $y$) under the \emph{supply rate function} $s:\R^{n_u}\times\R^{n_y}\rightarrow \R$ if there exists a \emph{storage function} $V: \R^{n_x} \rightarrow \R$ such that $V(0)=0$, $V(x)>0$ for all $x\neq \0$, and 
$$
\Delta V \triangleq V(x(t+1)) - V(x(t)) 
\leq  s(u(t),y(t)),
$$
for all $t\in\N^0$, all admissible input sequences, and all corresponding trajectories of \eqref{Eq:GeneralSystem}.
\end{definition}

The dissipativity property can be specialized by choosing different supply rate functions. Following \cite{WelikalaP52022}, we use a quadratic supply rate function determined by a coefficient matrix $\X$ to define a specialized dissipativity property named \emph{$\X$-dissipativity} as follows.

\begin{definition}\label{Def:X-Dissipativity}
The system \eqref{Eq:GeneralSystem} is \emph{$\X$-dissipative} where $\X = \X^\T \triangleq [\X^{kl}]_{k,l\in\N_2} \in\R^{(n_u+n_y)\times(n_u+n_y)}$ if it is dissipative under the supply rate function:
$$
s(u, y) \triangleq 
    \bm{u\\ y}^\T 
    \bm{\X^{11} & \X^{12}\\\X^{21} & \X^{22}}
    \bm{u\\y}.
$$
\end{definition}

The above-defined $\X$-dissipativity is equivalent to the conventional $(Q,S,R)$-dissipativity formulation \cite{Kottenstette2014}, and based on the used $\X$, it can represent several properties of interest as follows. 

\begin{remark}\label{Rm:X-DissipativityVersions}
If the system \eqref{Eq:GeneralSystem} is $\X$-dissipative with:  
\begin{enumerate}
\item $\X = \bm{\0 & \frac{1}{2}\I \\ \frac{1}{2}\I & \0}$, then it is \emph{passive};
\item $\X = \bm{-\nu\I & \frac{1}{2}\I \\ \frac{1}{2}\I & -\rho\I}$, then it is IF-OFP($\nu,\rho$), i.e., input feedforward and output feedback passive with indices $\nu$ and $\rho$, respectively \cite{WelikalaP42022};
\item $\X = \bm{\gamma^2\I & \0 \\ \0 & -\I}$, then it is L2G($\gamma$), i.e., finite-gain \emph{$L_2$-stable} with gain $\gamma$ \cite{WelikalaP42022};
\item $\X = \bm{-a\I & \frac{a+b}{2b}\I \\ \frac{a+b}{2b}\I & -\frac{1}{b}\I}$ (or $\X = \bm{-ab\I & \frac{a+b}{2}\I \\ \frac{a+b}{2}\I & -\I}$)  with $b>a$ and $b>0$ (or $b>a$), then it is sector bounded where $a$ and $b$ are sector bound parameters \cite{Kottenstette2014}.
\end{enumerate}  
\end{remark}

\subsection{Dissipativity of Linear Time-Invariant (LTI) Systems}\label{SubSec:dissipativity_LTI_systems}

If \eqref{Eq:GeneralSystem} is linear time-invariant (LTI), a necessary and sufficient condition for its $\X$-dissipativity can be found in the form of a linear matrix inequality (LMI) problem as follows.

\begin{proposition}\label{Pr:LTISystemXDisspativity}
\cite{Welikala2025Ax1} The linear time-invariant (LTI) system
\begin{equation}\label{Eq:Pr:LTISystemXDisspativity1}
\begin{aligned}
    x(t+1) &= A x(t) + B u(t),\\
    y(t) &= Cx(t) +Du(t),
\end{aligned}
\end{equation}
is $\X$-dissipative (from input $u$ to output $y$) if and only if there exists a matrix $P$ such that $P>0$ and 
\begin{equation}
\nonumber
\begin{aligned}
\label{Eq:Pr:LTISystemXDisspativity2}
\bm{
P & PA & PB \\
\star & P + C^\T \X^{22} C & C^\T \X^{21} + C^\T \X^{22} D\\
\star & \star & \X^{11} + \mathcal{H}(\X^{12} D) + D^\T \X^{22}D
}
\geq 0.
\end{aligned}
\end{equation}
\end{proposition}

Consider the LTI system \eqref{Eq:Pr:LTISystemXDisspativity1} with a local controller $u(t) \triangleq Kx(t) + \bar{u}(t)$ and $D\triangleq \0$, and let $\tilde{u}(t) \triangleq B\bar{u}(t)$ denote an external input (which may also include disturbances). The following corollary provides an LMI problem for designing this local controller (i.e., $K$) to enforce/optimize the corresponding closed-loop $\X$-dissipativity from $\tilde{u}(t)$ to $y(t)$. 

\begin{corollary}\label{Co:LTISystemXDisspativation}
\cite{Welikala2025Ax1} The closed-loop LTI system
\begin{equation}\label{Eq:Co:LTISystemXDisspativation1}
\begin{aligned}
x(t+1) &= (A+BK)x(t) + \tilde{u}(t),\\
y(t) &= Cx(t),   
\end{aligned}
\end{equation}
is $\X$-dissipative with $\X^{22}<0$ from external input $\tilde{u}(t)$ to output $y(t)$ if and only if there exists $P>0$ and $\bar{K}$ such that  
\begin{equation}
\nonumber 
\bm{ 
(-\X^{22})^{-1} & \0 & CP & \0 \\
\star & P & AP+B\bar{K} & \I \\
\star & \star & P  & PC^\top \X^{21}\\
\star & \star & \star & \X^{11}} \geq 0,
\end{equation}  
with $K=\bar{K}P^{-1}$.
\end{corollary}

\subsection{Networked Systems Design}\label{SubSec:NetworkedSystem}

Consider the networked system $\Sigma$ shown in Fig.~\ref{Fig:NetworkedSystem}, comprising $N$ independent discrete-time dynamic subsystems $\Sigma_i, i\in\N_N$ and a static interconnection matrix $M$ that defines how the subsystem inputs and outputs, and the networked system's exogenous inputs $w\in\R^{n_w}$ and performance outputs $z(t)\in\R^{n_z}$ are interconnected. 

\begin{figure}[!t]
    \centering
    \includegraphics[width=1.7in]{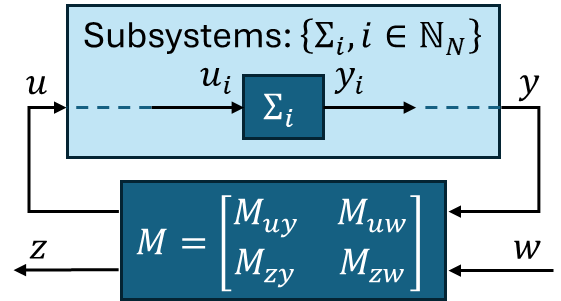}
	\caption{A generic networked system $\Sigma: w \rightarrow z$.}
	\label{Fig:NetworkedSystem}
\end{figure}

Each subsystem $\Sigma_i, i\in\N_N$ follows the dynamics
\begin{equation}\label{Eq:nonlinearSubsystemDynamics}
    \begin{aligned}
        x_i(t+1) &= f_i(x_i(t),u_i(t)),\\
        y_i(t) &= h_i(x_i(t),u_i(t)),
    \end{aligned}
\end{equation}
where the state $x_i(t)\in\R^{n_{xi}}$, input $u_i(t)\in\R^{n_{ui}}$, output $y_i(t)\in\R^{n_{yi}}$ and $t\in\N^0$. Further, under $u_i(t) = \0$, $x_i(t) = \0$ is an equilibrium point, i.e., $\0 = f_i(\0,\0)$, and $h_i(\0,\0) = \0$. Furthermore, it is $\X_i$-dissipative (from $u_i$ to $y_i$) where $\X_i = \X_i^\T \triangleq [\X_i^{kl}]_{k,l\in\N_2}$.

The interconnection matrix $M$ is structured according to the interconnection relationship:
\begin{align}
	\label{Eq:NSC2Interconnection}\bm{u\\z} = M \bm{y\\w} \equiv 
	\bm{M_{uy} & M_{uw} \\ M_{zy} & M_{zw}}\bm{y\\w},
\end{align}
where $u \triangleq [u_i^\T]_{i\in\N_N}^\T \in \R^{n_u}$ and $y \triangleq [y_i^\T]_{i\in\N_N}^\T \in \R^{n_y}$ with $n_u \triangleq \sum_{i\in\N_N} n_{ui}$ and $n_y \triangleq \sum_{i\in\N_N} n_{yi}$.

As shown in \cite{Welikala2025Ax1}, the interconnection matrix $M$ \eqref{Eq:NSC2Interconnection}, possibly together with selected entries of the desired network-level dissipativity matrix $\Y=\Y^\T\triangleq[\Y^{kl}]_{k,l\in\N_2}$, can be designed using an LMI problem to enforce and, when applicable, optimize the $\Y$-dissipativity of the networked system $\Sigma$ from $w$ to $z$. However, this requires two mild assumptions as detailed below.  

\begin{assumption}\label{As:NegativeDissipativity}
The desired $\Y$-dissipativity specification for the networked system $\Sigma$ is such that $\Y^{22}<0$. 
\end{assumption} 

\begin{remark}\label{Rm:As:NegativeDissipativity}
According to Rm. \ref{Rm:X-DissipativityVersions}, As. \ref{As:NegativeDissipativity} holds if we require the networked system $\Sigma$ to be L2G($\gamma$) (i.e., $L_2$-stable), or IF-OFP($\nu,\rho$) with $\rho>0$ (i.e., output feedback passive). As such properties are desirable, As. \ref{As:NegativeDissipativity} is mild.
\end{remark}

\begin{assumption}\label{As:PositiveDissipativity}
Each subsystem $\Sigma_i, i\in\N_N$ in the networked system $\Sigma$ is $\X_i$-dissipative with $\X_i^{11} > 0$. 
\end{assumption}

\begin{remark}\label{Rm:As:PositiveDissipativity}
According to Rm. \ref{Rm:X-DissipativityVersions}, As. \ref{As:PositiveDissipativity} holds if each subsystem $\Sigma_i,i\in\N_N$ is L2G($\gamma_i$) (i.e., $L_2$-stable), or IF-OFP($\nu_i,\rho_i$) with $\nu_i<0$ (i.e., lacks input feedback passivity). If such conditions are not met, local controllers can be used (e.g., see Co. \ref{Co:LTISystemXDisspativation}). Therefore, As.  \ref{As:PositiveDissipativity} is also mild. 
\end{remark}

\begin{proposition}\label{Pr:NSC2Synthesis}
\cite{Welikala2025Ax1} Under Assumptions \ref{As:NegativeDissipativity} and \ref{As:PositiveDissipativity}, the networked system $\Sigma$ is $\Y$-dissipative (from $w$ to $z$) if the interconnection matrix $M$ \eqref{Eq:NSC2Interconnection} is designed via the LMI problem:
\begin{equation}\label{Eq:Pr:NSC2Synthesis}
    \begin{aligned}
    \mbox{Find: }& L_{uy}, L_{uw}, M_{zy}, M_{zw}, \{p_i\in\R: i\in\N_N\}\\
    \mbox{Sub. to: }& p_i > 0, \forall i\in\N_N, \mbox{ and } \eqref{Eq:Pr:NSC2Synthesis2}
    \end{aligned}
\end{equation}
where 
$\X_p^{kl} \triangleq \diag([p_i\X_i^{kl}]_{i\in\N_N}), \forall k,l\in\N_2$, 
$\bar{\X}^{12} \triangleq \diag([(\X_i^{11})^{-1}\X_i^{12}]_{i\in\N_N})$, $\bar{\X}^{21} \triangleq (\bar{\X}^{12})^\T$ with 
$M_{uy} \triangleq (\X_p^{11})^{-1} L_{uy}$ and $M_{uw} \triangleq  (\X_p^{11})^{-1} L_{uw}$.
\end{proposition}

The above design technique for $M$ can also be used when $M$ is partially known. A specialized result for such a scenario is given in the following corollary. 

\begin{corollary}\label{Co:NSC2Synthesis}
Under Assumptions \ref{As:NegativeDissipativity} and \ref{As:PositiveDissipativity}, the networked system $\Sigma$ with $M_{uw} \triangleq \I, M_{zy} \triangleq \I$ and $M_{zw} \triangleq \0$, is $\Y$-dissipative (from $w$ to $z$) if the remaining interconnection matrix block $M_{uy}$ \eqref{Eq:NSC2Interconnection} is designed via the LMI problem:
\begin{equation}\label{Eq:Co:NSC2Synthesis}
\begin{aligned}
\mbox{Find: }& L_{uy}, \{p_i \in \R: i\in\N_N\}\\
\mbox{Sub. to: }& p_i > 0, \forall i\in\N_N, \mbox{ and }\\
\span \bm{
\X_p^{11} & \0 & L_{uy} & \X_p^{11} \\
\star & -\Y^{22} & -\Y^{22} & \0 \\ 
\star& \star & -\mathcal{H}(\bar{\X}^{21}L_{uy})-\X_p^{22} & -\X_p^{21}+\Y^{21} \\
\star & \star & \star &  \Y^{11}
} \geq 0.
\end{aligned}
\end{equation}
with $M_{uy} \triangleq (\X_p^{11})^{-1} L_{uy}$.
\end{corollary}
\begin{proof}
The proof follows substituting the known interconnection matrix blocks: $M_{uw} \triangleq \I$, $M_{zy} \triangleq \I$, and $M_{zw} \triangleq \0$ directly in \eqref{Eq:Pr:NSC2Synthesis2}, and then simplifying it using the relationships $M_{uw} \triangleq  (\X_p^{11})^{-1} L_{uw}$ and $\bar{\X}^{21}\X_p^{11} = \X_p^{21}$.
\end{proof}

\begin{figure*}[!t]
\begin{equation}\label{Eq:Pr:NSC2Synthesis2}
\bm{
\X_p^{11} & \0 & L_{uy} & L_{uw} \\
\star & -\Y^{22} & -\Y^{22}M_{zy} & -\Y^{22} M_{zw}\\ 
\star& \star & -\mathcal{H}(\bar{\X}^{21}L_{uy})-\X_p^{22} & -\bar{\X}^{21}L_{uw}+M_{zy}^\T\Y^{21} \\
\star & \star & \star &  \mathcal{H}(\Y^{12}M_{zw}) + \Y^{11}
} \geq  0
\end{equation}
\hrulefill
\end{figure*}

\section{Model-Based Hierarchical Design of Networked Systems}
\label{Sec:Model-Based}
In this section, we consider a special case of networked systems discussed in Sec. \ref{SubSec:NetworkedSystem} motivated by their occurrence across diverse applications, such as supply chains \cite{Welikala2025Ax1}, vehicular platoons \cite{WelikalaP72023}, and microgrids \cite{Najafirad2025P1}. Figure \ref{Fig:LinearNetworkedSystem} shows such a networked system, comprising subsystems that follow discrete-time LTI dynamics -- as opposed to generic non-linear dynamics \eqref{Eq:nonlinearSubsystemDynamics}, and hence referred to as a \emph{linear networked system}. 
The model-based design proceeds hierarchically: local controllers first enforce favorable subsystem-level $\X_i$-dissipativity certificates, which are then used to co-design the distributed global controller and the interconnection topology.

\begin{figure}[!t]
    \centering
    \includegraphics[width=2in]{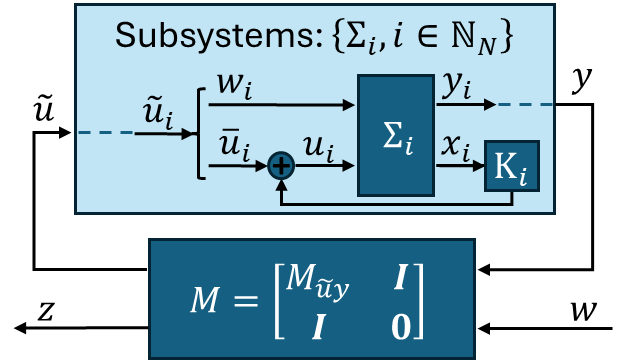}
	\caption{A linear networked system $\Sigma: w \rightarrow z$.}
	\label{Fig:LinearNetworkedSystem}
\end{figure}

\subsection{Local Controller Design}\label{SubSec:local_controller_design}
Consider a linear networked system (as shown in Fig. \ref{Fig:LinearNetworkedSystem}) where each subsystem $\Sigma_i, i\in\N_N$ follows the dynamics  
\begin{equation}\label{Eq:LinearSubsystemDynamics}
\begin{aligned}
x_i(t+1) =&\ A_ix_i(t) + B_iu_i(t) + w_i(t),\\
y_i(t) =&\ C_ix_i(t). 
\end{aligned}
\end{equation}
In \eqref{Eq:LinearSubsystemDynamics}, $w_i(t)$ represents the local disturbance component in $w \triangleq [w_i^\T]_{i\in\N_N}^\T$ (i.e., the exogenous input affecting the networked system). At each subsystem $\Sigma_i, i\in\N_N$, we also include a local state-feedback controller of the form: 
\begin{equation}\label{Eq:SubsystemController}
u_i(t) = K_ix_i(t) + \bar{u}_i(t),    
\end{equation}
where $\bar{u}_i(t)$ is the input generated by the distributed global controller to be designed in the next subsection. Applying \eqref{Eq:SubsystemController} in \eqref{Eq:LinearSubsystemDynamics}, we get the closed-loop subsystem dynamics
\begin{equation}\label{Eq:ClosedLoopSubsystemDynamics}
\begin{aligned}
x_i(t+1) =&\ (A_i+B_i K_i)x_i(t) + \tilde{u}_i(t),\\
y_i(t) =&\ C_ix_i(t). 
\end{aligned}
\end{equation}
where $\tilde{u}_i(t) \triangleq B_i \bar{u}_i(t) + w_i(t)$. 

The following proposition shows how the local controller $K_i$ in \eqref{Eq:SubsystemController} can be designed so that the closed-loop subsystem \eqref{Eq:ClosedLoopSubsystemDynamics} is $\X_i$-dissipative from $\tilde{u}_i$ to $y_i$. The resulting local dissipativity certificate is constrained to be compatible with the subsequent network-level design. In particular, $\X_i^{11}>0$ is enforced as required by As. \ref{As:PositiveDissipativity}, while $\X_i^{22}<0$ is imposed to be compatible with dissipativity classes used in Co. \ref{Co:LTISystemXDisspativation}.

\begin{proposition}\label{Pr:LocalControllerDesign}
At each subsystem $\Sigma_i, i\in\N_N$, the local controller $K_i$ in \eqref{Eq:SubsystemController} can be designed to enforce the closed-loop $\X_i$-dissipativity of \eqref{Eq:ClosedLoopSubsystemDynamics} from $\tilde{u}_i$ to $y_i$, with $\X_i^{11}>0$ and $\X_i^{22}<0$, using the LMI problem:
\begin{equation}\label{Eq:Pr:LocalControllerDesign}
\begin{aligned}
\mbox{Find: } &\bar{K}_i, P_i, \X_i\\
\mbox{Sub. to: } &P_i > 0,\ \X_i^{11}>0,\ \X_i^{22}<0,\  \mbox{ and }\\
\span \bm{ 
(-\X_i^{22})^{-1} & \0 & C_iP_i & \0 \\
\star & P_i & A_iP_i+B_i\bar{K}_i & \I \\
\star & \star & P_i  & P_iC_i^\T \X_i^{21}\\
\star & \star & \star & \X^{11}_i } \geq 0, 
\end{aligned}
\end{equation}
with $K_i=\bar{K}_iP_i^{-1}$.
\end{proposition}
\begin{proof}
Comparing \eqref{Eq:ClosedLoopSubsystemDynamics} with \eqref{Eq:Co:LTISystemXDisspativation1}, it is clear that Co. \ref{Co:LTISystemXDisspativation} is applicable, which leads to the given LMI problem \eqref{Eq:Pr:LocalControllerDesign}. 
\end{proof}

\begin{remark}\label{Rm:LocalXiOptimization}
The constraint $\X_i^{22}<0$ is consistent with common dissipativity specifications such as finite-gain $L_2$ stability and output-feedback passivity, as discussed in Rm. \ref{Rm:As:NegativeDissipativity}. The LMI in Prop. \ref{Pr:LocalControllerDesign} can be used either to certify a prescribed $\X_i$ or to optimize selected blocks of $\X_i$ when the resulting constraints remain convex. In particular, if one wishes to optimize entries of $\X_i^{22}$ appearing through $(-\X_i^{22})^{-1}$, an extra change of variables or a fixed parametrization of $\X_i^{22}$ should be used to preserve the LMI structure.
\end{remark}

\begin{remark}
If each subsystem $\Sigma_i, i\in\N_N$ followed analogous continuous-time dynamics, then the main LMI condition in \eqref{Eq:Pr:LocalControllerDesign} has to be replaced with
$$
\bm{ 
(-\X_i^{22})^{-1} & C_iP_i & \0 \\
\star & -\mathcal{H}(A_iP_i+B_i\bar{K}_i) & -\I + P_iC_i^\T \X_i^{21} \\
\star & \star & \X^{11}_i } \geq 0.
$$
\end{remark}

\subsection{Distributed Global Controller Design}

At each subsystem $\Sigma_i, i\in\N_N$, we consider a distributed global controller $\bar{u}_i(t)$ (see \eqref{Eq:SubsystemController}) of the form:
\begin{equation}\label{Eq:DistributedGlobalController}
\bar{u}_i(t) = \sum_{j \in \N_N} K_{ij} y_j(t) \iff \bar{u}(t) = K y(t) 
\end{equation}
where the equation in the RHS is from vectorizing the equation in the LHS, with the notations $\bar{u} \triangleq [\bar{u}_i^\T]_{i\in\N_N}^\T$, $y \triangleq [y_i^\T]_{i\in\N_N}^\T$ and $K\triangleq [K_{ij}]_{i,j\in\N_N}$. 
Note that, in \eqref{Eq:DistributedGlobalController}, $K_{ij}=\0$ if the output information from the subsystem $\Sigma_j$ is not required at the subsystem $\Sigma_i$ (implying no communication link is necessary from $\Sigma_j$ to $\Sigma_i$). Therefore, while designing the distributed global controller $K$ \eqref{Eq:DistributedGlobalController}, we can also optimize the underlying communication topology. Consequently, this distributed global controller design stage is referred to as a communication and control co-design stage \cite{WelikalaP72023}. 

At this stage, the local controllers \(K_i\) and the corresponding dissipativity certificates \(\X_i\) obtained from the local design are treated as fixed. For this co-design task, we can apply Prop. \ref{Pr:NSC2Synthesis} exploiting the subsystem dissipativity properties enforced by the local controller design \eqref{Eq:Pr:LocalControllerDesign}. However, we first need to identify the underlying interconnection matrix structure of the networked system under consideration. To this end, recall the external input affecting the closed-loop subsystems, i.e., $\tilde{u}_i(t)$ in \eqref{Eq:ClosedLoopSubsystemDynamics}, that can be vectorized and simplified (using \eqref{Eq:DistributedGlobalController}) respectively as
\begin{equation} \label{Eq:ExternalInputOnSubsystems}
\begin{aligned}
\tilde{u}_i(t) \triangleq B_i \bar{u}_i(t) + w_i(t) 
&\iff \tilde{u}(t) = B\bar{u}(t) + w(t) \\ 
&\iff \tilde{u}(t) = BKy(t) + w(t),
\end{aligned}
\end{equation}
where $\tilde{u} \triangleq [\tilde{u}_i^\T]_{i\in\N_N}^\T$ and $B \triangleq \diag([B_i]_{i\in\N_N})$. On the other hand, let us define $z_i(t) \triangleq y_i(t)$ as the local performance output component in $z \triangleq [z_i^\T]^\T_{i\in\N_N}$. Therefore, the performance output of the networked system is $z(t) \triangleq y(t)$. Combining this result with \eqref{Eq:ExternalInputOnSubsystems}, we get the underlying interconnection matrix structure through the interconnection relationship  
\begin{equation}\label{Eq:InterconnectionMatrixLinearNetworkedSystem}
\bm{\tilde{u}\\z} = \bm{M_{uy} & M_{uw} \\ M_{zy} & M_{zw}} \bm{y\\w} \equiv \bm{BK & \I \\ \I & \0}\bm{y\\w}. 
\end{equation}

Based on \eqref{Eq:InterconnectionMatrixLinearNetworkedSystem}, it is clear that we can use Co. \ref{Co:NSC2Synthesis} (instead of Prop. \ref{Pr:NSC2Synthesis}) for this co-design task. The details are summarized in the following proposition. 

\begin{proposition}\label{Pr:GlobalControllerDesign}
Suppose that the local controllers \eqref{Eq:SubsystemController} have been designed (e.g., via \eqref{Eq:Pr:LocalControllerDesign}) so that each closed-loop subsystem \eqref{Eq:ClosedLoopSubsystemDynamics} is $\X_i$-dissipative (from $\tilde{u}_i$ to $y_i$, with $\X_i^{11}>0$, i.e., As. \ref{As:PositiveDissipativity}), and suppose that the desired network-level dissipativity certificate satisfies $\Y^{22}<0$ (i.e., As. \ref{As:NegativeDissipativity}). Then, the distributed global controller and the underlying communication topology  (jointly captured by $K$ in \eqref{Eq:DistributedGlobalController}) can be co-designed to enforce/optimize the $\Y$-dissipativity (from $w$ to $z$) of the closed-loop linear networked system using the LMI problem:  
\begin{equation}\label{Eq:Pr:GlobalControllerDesign}
\begin{aligned}
\min_{\substack{\bar{K}, \Y, \\ \{p_i\in\R: i\in\N_N\}}}&\ J \triangleq \Vert \bar{K} \Vert_1 - \phi([p_i]_{i\in\N_N}) + \psi(\Y)\\
\mbox{Sub. to:}&\ p_i > 0, \forall i\in\N_N,\ \Y^{22}<0,\ \mbox{ and }\\
\span \bm{
\X_p^{11} & \0 & \X^{11} B \bar{K} & \X_p^{11} \\
\star & -\Y^{22} & -\Y^{22} & \0 \\ 
\star& \star & -\mathcal{H}(\X^{21} B \bar{K})-\X_p^{22} & -\X_p^{21}+\Y^{21} \\
\star & \star & \star &  \Y^{11}
} \geq 0,
\end{aligned}    
\end{equation}
where $\X_p^{kl} \triangleq \diag([p_i\X_i^{kl}]_{i\in\N_N})$, $\X^{kl} \triangleq \diag([\X^{kl}_i]_{i\in\N_N})$, $\forall k,l\in\N_2$, $\Y \triangleq [\Y^{kl}]_{k,l\in\N_2}$, and $\bar{K} \triangleq [\bar{K}_{ij}]_{i,j\in\N_N}$ with $K \triangleq [p_i^{-1} \bar{K}_{ij}]_{i,j\in\N_N}$ (functions $J$, $\phi$, and $\psi$ are interpreted as in Rm. \ref{Rm:GlobalObjective}).
\end{proposition}
\begin{proof}
We start by applying Co. \ref{Co:NSC2Synthesis} with $M_{uy} = BK$ from \eqref{Eq:InterconnectionMatrixLinearNetworkedSystem}. According to Co. \ref{Co:NSC2Synthesis}, $M_{uy} = (\X_p^{11})^{-1} L_{uy}$. Therefore, we have to replace the $L_{uy}$ term appearing in the main LMI condition in Co. \ref{Co:NSC2Synthesis} using $L_{uy} = \X_p^{11} BK$. Here, since the RHS is bilinear in decision variables $\{p_i\in\R: i\in\N_N\}$ and $K$, we use a change of variables: $\bar{K} \triangleq [\bar{K}_{ij}]_{i,j\in\N_N} \equiv [p_i K_{ij}]_{i,j\in\N_N}$. Consequently, we get: 
\begin{equation*}
\begin{aligned}
L_{uy} =&\ \X_p^{11} B K \\
&\iff L_{uy}^{ij} = p_i \X_i^{11} B_i K_{ij} = \X_i^{11} B_i \bar{K}_{ij}, \forall i,j \in \N_N \\ 
&\iff L_{uy} = \X^{11} B \bar{K}.        
\end{aligned}
\end{equation*}
Using the above derived substitution for $L_{uy}$ with the fact 
\begin{equation*}
\begin{aligned}
\bar{\X}^{21}\X^{11} =&\ \diag([\X_i^{21}(\X_i^{11})^{-1}]_{i\in\N_N}) \diag([\X_i^{11}]_{i\in\N_N})\\
=&\ \diag([\X_i^{21}]_{i\in\N_N}) = \X^{21}
\end{aligned}
\end{equation*}
we get the LMI condition in \eqref{Eq:Pr:GlobalControllerDesign}.  
\end{proof}

\begin{remark}\label{Rm:GlobalObjective}
Unlike in Co. \ref{Co:NSC2Synthesis}, Prop. \ref{Pr:GlobalControllerDesign} optimizes an objective function $J$ in addition to enforcing feasibility. The term $\Vert \bar K\Vert_1$ promotes sparsity of $\bar K$, and therefore sparsity of the communication topology induced by $K$. The term $-\phi([p_i]_{i\in\N_N})$ can be used to encourage larger scaling variables $p_i$, which tends to reduce the recovered gains $K_{ij}=p_i^{-1}\bar K_{ij}$. For example, one may choose $\phi(p)=\alpha\sum_i p_i$ with $\alpha>0$, together with upper bounds on $p_i$ if needed. The term $\psi(\Y)$ is used to optimize the desired network-level performance certificate, such as passivity indices or an $L_2$-gain bound; for instance, $\psi(\Y)=\beta\gamma^2$ when $\Y$ is parameterized by an $L_2$-gain level $\gamma$. The functions $\phi$ and $\psi$, together with any additional bounds or parametrizations, should be selected so that $J$ is convex whenever a convex optimization problem is desired.
\end{remark}


\subsection{Local Controller Design Revisit}

The global co-design problem \eqref{Eq:Pr:GlobalControllerDesign}, particularly its feasibility and effectiveness, is clearly dependent on the subsystem dissipativity properties enforced by the local controller design problem \eqref{Eq:Pr:LocalControllerDesign}. Therefore, we propose to influence the local controller design \eqref{Eq:Pr:LocalControllerDesign} to achieve favorable subsystem dissipativity properties that can potentially lead to feasible and effective global co-designs. The main idea is to identify a necessary condition for the LMI conditions in the global co-design problem \eqref{Eq:Pr:GlobalControllerDesign}, and then to decompose it into a form that can be included in the local controller design problems \eqref{Eq:Pr:LocalControllerDesign} at different subsystems.

Using Lm. \ref{Lm:BEW}, it is easy to see that the main LMI in \eqref{Eq:Pr:GlobalControllerDesign} is equivalent to $[W_{ij}]_{i,j\in\N_N} \geq 0$, where $W_{ij}$ is as defined in \eqref{Eq:NecessaryCondStep1}. Now, a set of necessary conditions for $[W_{ij}]_{i,j\in\N_N} \geq 0$ can be identified as $\{W_{ii} \geq 0, \forall i\in\N_N\}$. Dividing these necessary conditions by the respective scalars $\{p_i: i\in\N_N\}$ and defining $\bar{\Y}^{kl}_{ii} \triangleq p_i^{-1}\Y^{kl}_{ii}, \forall k,l\in\N_2, i\in\N_N$, we get the necessary conditions as $p_i^{-1}W_{ii} \geq 0$, i.e., \eqref{Eq:NecessaryCondStep2}, $\forall i\in\N_N$.

\begin{figure*}
\begin{equation}\label{Eq:NecessaryCondStep1}
W_{ij} \triangleq 
\bm{
p_i\X_i^{11}\e_{ij} & \0 & \X_i^{11} B_i \bar{K}_{ij} & p_i\X_i^{11}\e_{ij} \\
\star & -\Y^{22}_{ij} & -\Y^{22}_{ij} & \0 \\ 
\star& \star & -\X_i^{21}B_i\bar K_{ij}-(\X_j^{21}B_j\bar K_{ji})^\T-p_i\X_i^{22}\e_{ij} & -p_i\X_i^{21}\e_{ij}+\Y^{21}_{ij} \\
\star & \star & \star &  \Y^{11}_{ij}
}
\end{equation}
\begin{equation}\label{Eq:NecessaryCondStep2}
\bm{
\X_i^{11} & \0 & p_i^{-1}\X_i^{11} B_i \bar{K}_{ii} & \X_i^{11} \\
\star & -p_i^{-1}\Y^{22}_{ii} & -p_i^{-1}\Y^{22}_{ii} & \0 \\ 
\star& \star & -\mathcal{H}(p_i^{-1}\X^{21}_i B_i \bar{K}_{ii})-\X_i^{22} & -\X_i^{21}+ p_i^{-1} \Y^{21}_{ii} \\
\star & \star & \star &  p_i^{-1}\Y^{11}_{ii}
} \geq 0  
\end{equation} 
\begin{equation}\label{Eq:NecessaryCondStep3}
\bm{
(-\x_i^{22})\x_i^{11}\I & \0 & p_i^{-1}\x_i^{11}B_i \bar{K}_{ii} & (-\x_i^{22})\x_i^{11}\I \\
\star & -p_i^{-1}(-\x_i^{22})^{-1}\y^{22}_{ii}\I & -p_i^{-1}(-\x_i^{22})^{-1}\y^{22}_{ii}\I & \0 \\ 
\star& \star & -\mathcal{H}(p_i^{-1}(-\x_i^{22})^{-1}\x^{21}_i B_i \bar{K}_{ii}) + \I & -\x_i^{21}\I+ p_i^{-1}\y^{21}_{ii}\I \\
\star & \star & \star &  p_i^{-1}(-\x_i^{22})\y^{11}_{ii}\I
} \geq 0 
\end{equation}
\hrulefill
\end{figure*}

To further simplify these necessary conditions, we make the following technical assumption on the dissipativity properties of the desired subsystem and the networked system.

\begin{assumption}\label{As:ScalarDissipativityMatrices}
The desired subsystem and networked system dissipativity properties, i.e., $\{\X_i: i\in\N_N\}$ and $\Y\triangleq[[\Y_{ij}^{kl}]_{i,j\in\N_N}]_{k,l\in\N_2}$ matrices, respectively, have inner blocks that are scalar matrices such that $\X_i^{kl} = \x_i^{kl}\I$ and $\Y_{ii}^{kl} = \y_{ii}^{kl}\I$, for any $k,l \in \N_2, i\in\N_N$. Further, the scalars corresponding to input-output couplings, i.e., $\x_i^{kl}$ and $\y_{ii}^{kl}$ for any $k \neq l$, are known. 
\end{assumption}

\begin{remark}
Assumption \ref{As:ScalarDissipativityMatrices} holds for the first three dissipativity properties given in Rm. \ref{Rm:X-DissipativityVersions}. This includes the widely used case where each subsystem $\Sigma_i, i\in\N_N$ is enforced to be IF-OFP($\nu_i,\rho_i$) while the closed-loop networked system is enforced to be L2G($\gamma$) (or any similar combination). If the desired subsystem and networked system dissipativity properties violate this assumption, an alternative set of necessary conditions can be derived from \eqref{Eq:NecessaryCondStep2} by following similar steps, as detailed below. 
\end{remark}
Under As. \ref{As:ScalarDissipativityMatrices}, applying Lm. \ref{Lm:Congruence} to \eqref{Eq:NecessaryCondStep2} with \(P\triangleq\diag((-\x_i^{22})\I,\I,\I,(-\x_i^{22})\I)\), and then multiplying the resulting inequality by the positive scalar \((-\x_i^{22})^{-1}\), gives \eqref{Eq:NecessaryCondStep3}.
Next, via the change of variables: 
\begin{equation}\label{Eq:ChangeOfVariables}
\begin{gathered}
\bar{\y}_{ii}^{11} \triangleq p_i^{-1}(-\x_i^{22})\y^{11}_{ii}, \quad \quad
\bar{\y}_{ii}^{12} \triangleq \bar{\y}_{ii}^{21} \triangleq p_i^{-1}\y^{21}_{ii},\\
\bar{\y}_{ii}^{22} \triangleq p_i^{-1}(-\x_i^{22})^{-1}\y^{22}_{ii}, \quad \quad
\tilde{K}_{ii} \triangleq p_i^{-1}\x_i^{11} \bar{K}_{ii}, \\
\bar{\x}_i^{11} \triangleq (-\x_i^{22})\x_i^{11},
\end{gathered}
\end{equation}
\eqref{Eq:NecessaryCondStep3} can be expressed as 
\begin{equation}\label{Eq:NecessaryConditions}
\begin{aligned}
\bm{
\bar{\x}_i^{11}\I & \0 & B_i \tilde{K}_{ii} & \bar{\x}_i^{11}\I \\
\star & -\bar{\y}_{ii}^{22}\I & -\bar{\y}_{ii}^{22}\I & \0 \\ 
\star& \star & -\mathcal{H}(B_i\hat{K}_{ii}) + \I & -\x_i^{21}\I+ \bar{\y}_{ii}^{21}\I \\
\star & \star & \star &  \bar{\y}_{ii}^{11}\I
} \\ \geq 0 
\end{aligned}
\end{equation}
This condition is obtained by relaxing the nonlinear equality $\hat{K}_{ii}=\x_i^{21}(\bar{\x}_i^{11})^{-1}\tilde K_{ii}$, which is equivalent to 
$\hat{K}_{ii}=\theta_i\tilde K_{ii}$ with $\theta_i \triangleq \x^{21}_i(\bar{\x}_i^{11})^{-1}$. Therefore, the relaxed LMIs \eqref{Eq:NecessaryConditions}, for all $i\in\N_N$, provide a tractable set of necessary conditions associated with the global co-design problem \eqref{Eq:Pr:GlobalControllerDesign}.

Note that the relaxed condition \eqref{Eq:NecessaryConditions} is an LMI in the local variables
$\bar{\x}_i^{11}$ (a transformed subsystem dissipativity variable), $\bar{\y}_{ii}^{11},\bar{\y}_{ii}^{21},\bar{\y}_{ii}^{22}$ (transformed diagonal blocks of the desired network-level dissipativity matrix), and $\tilde K_{ii},\hat K_{ii}$ (transformed local components of the global controller). The original subsystem parameter $\x_i^{11}$ is not optimized directly in \eqref{Eq:NecessaryConditions}; instead, it will be recovered after the local design using $\x_i^{11}=(-\x_i^{22})^{-1}\bar{\x}_i^{11}$, where $\x_i^{22}<0$ is optimized through the local dissipativity LMI. 
Similarly, the original quantities $p_i$, $\Y_{ii}^{kl}$, and $\bar{K}_{ii}$ are not required to be recovered at this stage; when needed, they can be reconstructed from \eqref{Eq:ChangeOfVariables} for interpretation, initialization, or objective-function design. Therefore, \eqref{Eq:NecessaryConditions} can be inserted into the local controller design problem as an additional convex constraint that biases the local certificate toward improving the feasibility and effectiveness of the subsequent global co-design problem.

Using the obtained LMI variables and the used change of variables relationships, while it appears that one can ``recover'' some of the global co-design variables (e.g., when $\y_{ii}^{12}=\y_{ii}^{21}=0$ and $\y_{ii}^{22}$ is known, $p_i = (\bar{\y}_{ii}^{22})^{-1} (-\x_i^{22})^{-1} \y_{ii}^{22}$,
$\y_{ii}^{11} = p_i(-\x_i^{22})^{-1}\bar{\y}_{ii}^{11}$, and 
$\bar{K}_{ii} = p_i(\x_i^{11})^{-1}\tilde{K}_{ii}$), such ``recovered'' variables are neither required nor expected to align perfectly with the subsequent corresponding global co-design variable values observed when \eqref{Eq:Pr:GlobalControllerDesign} is solved. After all, \eqref{Eq:NecessaryConditions} is only necessary but not sufficient for \eqref{Eq:Pr:GlobalControllerDesign}.

\begin{remark}
When each subsystem $\Sigma_i$ has an input matrix $B_i$ structured in a way that the matrix block $B_i \tilde{K}_{ii}$ in \eqref{Eq:NecessaryConditions} has at least one zero in its main diagonal, we can reapply Lm. \ref{Lm:BEW} to identify a set of necessary conditions for  \eqref{Eq:Pr:GlobalControllerDesign} as
\begin{equation}
\label{Eq:NecessaryConditions3}
\begin{aligned}
\bm{
\bar{\x}_i^{11} & 0 & 0 & \bar{\x}_i^{11} \\
\star & -\bar{\y}^{22}_{ii} & -\bar{\y}^{22}_{ii} & \0 \\ 
\star& \star & 1 & -\x_i^{21}+\bar{\y}^{21}_{ii} \\
\star & \star & \star &  \bar{\y}^{11}_{ii} 
} \geq 0,  
\forall i\in\N_N.
\end{aligned}
\end{equation}
\end{remark}

The following proposition combines the obtained necessary conditions \eqref{Eq:NecessaryConditions} with the respective local LMI problems \eqref{Eq:Pr:LocalControllerDesign} to formulate a comprehensive local controller design problem that positively impacts the feasibility and the effectiveness of the subsequent global co-design problem \eqref{Eq:Pr:GlobalControllerDesign}.

\begin{proposition}\label{Pr:LocalControllerDesignRevisit}
At each subsystem $\Sigma_i, i\in\N_N$, under As.  \ref{As:ScalarDissipativityMatrices}, to enforce/optimize the closed-loop subsystem $\X_i$-dissipativity in \eqref{Eq:ClosedLoopSubsystemDynamics} (from $\tilde{u}_i$ to $y_i$, where $\X_i \triangleq [\x_i^{kl}\I]_{k,l\in\N_2}$ with $\x_i^{11} > 0$ and $\x_i^{22} < 0$) so that it favors the subsequent global co-design stage (i.e., \eqref{Eq:Pr:GlobalControllerDesign} in Prop. \ref{Pr:GlobalControllerDesign}), the local controller $K_i$ \eqref{Eq:SubsystemController} can be designed using the LMI problem: 
\begin{equation}\label{Eq:Pr:LocalControllerDesignRevisit}
\begin{aligned}
\mbox{Find: } &\tilde{K}_i, \bar{P}_i, \{\bar{\x}_i^{11}, \x_i^{22}\}, \\ 
&\{\tilde{K}_{ii}, \hat{K}_{ii}, \bar{\y}_{ii}^{11}, \bar{\y}_{ii}^{21}, \bar{\y}_{ii}^{22}\}\\
\mbox{Sub. to: } &\bar{P}_i > 0,\ \bar{\x}_i^{11}>0,\ \x_i^{22} < 0,\ \eqref{Eq:NecessaryConditions},\ \bar{\y}^{22}_{ii} < 0,\ \mbox{ and }\\
\span \bm{ 
\I & \0 & C_i \bar{P}_i & \0 \\
\star & \bar{P}_i & A_i\bar{P}_i+B_i\tilde{K}_i & (-\x_i^{22})\I \\
\star & \star & \bar{P}_i  & \x^{21}_i \bar{P}_i C_i^\T \\
\star & \star & \star & \bar{\x}^{11}_i\I 
} \geq 0,
\end{aligned}
\end{equation}
with $K_i=\tilde{K}_i \bar{P}_i^{-1}$ and $\x_i^{11}=(-\x_i^{22})^{-1}\bar{\x}_i^{11}$. The scalars $\x_i^{12}$ and $\x_i^{21}$ are known due to As. \ref{As:ScalarDissipativityMatrices}. The remaining global co-design variables, such as $p_i$, $\Y_{ii}^{kl}$, and $\bar K_{ii}$, need not be recovered for this local design step, but can be reconstructed from \eqref{Eq:ChangeOfVariables} when needed.
\end{proposition}

\begin{proof}
Applying As. \ref{As:ScalarDissipativityMatrices} in \eqref{Eq:Pr:LocalControllerDesign} and then multiplying its LMIs by $(-\x_i^{22})$, 
we obtain the constraints \(\bar P_i>0\), \(\bar{\x}_i^{11}>0\), \(\x_i^{22}<0\), and the final LMI in \eqref{Eq:Pr:LocalControllerDesignRevisit}, with the change of variables: $\bar{\x}_i^{11}\triangleq (-\x_i^{22})\x_i^{11}$, $\bar{P}_i \triangleq (-\x_i^{22}) P_i$, and $\tilde{K}_i \triangleq (-\x_i^{22})\bar{K}_i$.

From the preceding derivation, \eqref{Eq:NecessaryConditions} provides a relaxed necessary condition for the feasibility of the global co-design stage \eqref{Eq:Pr:GlobalControllerDesign}. After the relaxation 
$\hat K_{ii}=\theta_i\tilde K_{ii}$, this condition is an LMI in the variables 
$\{\bar{\x}_i^{11},\bar{\y}_{ii}^{11},\bar{\y}_{ii}^{21},\bar{\y}_{ii}^{22},\tilde K_{ii},\hat K_{ii}\}$. 
Thus, it can be imposed directly in \eqref{Eq:Pr:LocalControllerDesignRevisit}. The additional constraint $\bar{\y}_{ii}^{22}<0$ is included so that, whenever the corresponding global variable $p_i$ is reconstructed through \eqref{Eq:ChangeOfVariables} under $\y_{ii}^{22}<0$ and $\x_i^{22}<0$, the recovered value satisfies $p_i>0$. This completes the proof.
\end{proof}

\begin{remark}
In the revised local controller design problem \eqref{Eq:Pr:LocalControllerDesignRevisit}, one can include an objective function (in addition to finding a feasible solution) to maximize the influence on the feasibility and effectiveness of the global co-design problem. To see this, consider the scenario where 
$\y_{ii}^{12}=\y_{ii}^{21}=0$ and $\y_{ii}^{22}$ is known - regarding the global co-design specifications. For this scenario, we saw that the global co-design problem variables can be recovered from \eqref{Eq:Pr:LocalControllerDesignRevisit} as 
$p_i = (\bar{\y}_{ii}^{22})^{-1}(-\x_i^{22})^{-1} \y_{ii}^{22}$, 
$\y_{ii}^{11} = p_i(-\x_i^{22})^{-1}\bar{\y}_{ii}^{11}$, 
and 
$\bar{K}_{ii} = p_i(\x_i^{11})^{-1}\tilde{K}_{ii}$. Therefore, to optimize the global co-design problem variables $p_i, \y^{11}_{ii}$ and $\bar{K}_{ii}$, the local design variables $\bar{\y}_{ii}^{22}, \bar{\y}_{ii}^{11}$, and $\tilde K_{ii}$ can be optimized, respectively.
\end{remark}

\begin{remark}
In the revised local controller design problem \eqref{Eq:Pr:LocalControllerDesignRevisit}, we can include the bilinear matrix inequality constraint:
$$
\bar{\x}_i^{11}\hat{K}_{ii} = \x_i^{21}\tilde{K}_{ii} 
$$
to undo the relaxation $\hat{K}_{ii} = \theta_i\tilde{K}_{ii}$ introduced in \eqref{Eq:NecessaryConditions}. This will make \eqref{Eq:NecessaryConditions} more consistent/aligned with the global co-design problem \eqref{Eq:Pr:GlobalControllerDesign}. Alternatively, we can keep the term $(\theta_i\tilde{K}_{ii})$ in \eqref{Eq:NecessaryConditions} and treat $\theta_i$ as a design parameter - to be prespecified or tuned using a 1-D search. 
\end{remark}

\begin{remark}
If each subsystem $\Sigma_i, i\in\N_N$ follows analogous continuous-time dynamics, the final LMI condition in  \eqref{Eq:Pr:LocalControllerDesignRevisit} has to be replaced with
\begin{equation}
\bm{
\I & C_i\bar{P}_i & \0 \\
\star & -\mathcal{H}(A_i \bar{P}_i + B_i \tilde{K}_i) & \x_i^{22}\I +  \x^{21}_i \bar{P}_i C_i^\T \\
\star & \star & \bar{\x}^{11}_i \I 
}  \geq 0.
\end{equation}
\end{remark}

\section{Data-Driven Hierarchical Design of Networked Systems}
\label{Sec:Data-Driven}

In this section, for the class of linear networked systems considered (see Fig. \ref{Fig:LinearNetworkedSystem}), we present a data-driven hierarchical control design framework. In particular, we exploit the model-based LMI formulations developed for local controller design \eqref{Eq:Pr:LocalControllerDesignRevisit} and global co-design \eqref{Eq:Pr:GlobalControllerDesign} in Sec. \ref{Sec:Model-Based}, but strategically replace the role of subsystem model parameters by modifying the overall design framework using data trajectories obtained from the subsystems. 

Overall, the proposed approach follows a direct data-driven control design strategy: the local and global controllers are synthesized directly from observed subsystem trajectory data, without explicitly identifying subsystem model parameters. This avoids a separate model-identification step and naturally leads to robust designs over all models consistent with the data and disturbance bounds. Moreover, it leads to robust networked system designs that can be conveniently extended to integrate online design, fault detection, and recovery mechanisms. 

\subsection{Data-Driven Local Controller Design}

While each subsystem $\Sigma_i, i\in\N_N$ follows the dynamic model \eqref{Eq:LinearSubsystemDynamics}, here we assume that we do not know its exact model parameters $\Theta_i \equiv (A_i, B_i, C_i)$ in \eqref{Eq:LinearSubsystemDynamics}, but have finite data trajectories $\Phi_i \triangleq \{(u_i(t),x_i(t),y_i(t)): t\in\N_T^0\}$ at our disposal. For each subsystem $\Sigma_i, i\in\N_N$, let us define the data matrices (as $1 \times T$ block row matrices):
\begin{equation}\label{Eq:DataMatrices}
\begin{aligned}
X_i \triangleq [x_i(t)]_{t\in\N_{T-1}^0},& \quad  
\bar{X}_i \triangleq [x_i(t+1)]_{t\in\N_{T-1}^0},\\ 
U_i \triangleq [u_i(t)]_{t\in\N_{T-1}^0},& \quad
Y_i \triangleq [y_i(t)]_{t\in\N_{T-1}^0},\\ 
\span W_i \triangleq [w_i(t)]_{t\in\N_{T-1}^0}.
\end{aligned}
\end{equation}
Note that, while we know the ``system'' data matrices $\bar{\Phi}_i \triangleq  (X_i, \bar{X}_i, U_i, Y_i),\ i\in\N_N$, we do not know ``disturbance'' data matrices $W_i, i\in\N_N$. However, we make the following assumption about such disturbance data matrices. 

\begin{assumption}\label{As:DisturbanceData}
For each subsystem $\Sigma_i, i\in\N_N$, its disturbance data matrix $W_i$ \eqref{Eq:DataMatrices} is bounded such that
\begin{equation}\label{Eq:As:DisturbanceData}
\bm{\I \\ W_i^\T}^\T Q_{wi} \bm{\I  \\ W_i^\T} \geq 0
\end{equation} 
where $Q_{wi} \triangleq [Q_{wi}^{kl}]_{k,l\in\N_2}$ is symmetric and $Q_{wi}^{22}<0$.
\end{assumption}

\begin{remark}
The above As. \ref{As:DisturbanceData} is a standard assumption used in data-driven control literature \cite{Waarde2022}. It represents an elliptical bound on disturbances, and hence covers simple spherical bounds on disturbances like $W_i W_i^\T \leq \gamma_i \I$ where $\gamma_{i} > 0$ is a constant. 
In some instances (e.g., with spherical bounds on disturbances), when a family of admissible disturbance bounds is available, \(Q_{wi}\) may be selected or tuned within that family to reduce conservatism while preserving the validity of the bound. 
\end{remark}
 
We also impose the following standard data-richness condition on the data $\Phi_i$, which is closely related to the persistence of excitation in data-driven control \cite{de2019formulas}.

\begin{assumption}\label{As:PersistenceOfExcitation}
For each subsystem $\Sigma_i, i\in\N_N$, 
$$
\text{rank} \bm{ U_i \\ X_i } = n_{ui} + n_{xi} \quad \mbox{(i.e., full row rank)}.
$$
\end{assumption}

For each subsystem $\Sigma_i, i\in\N_N$, as the data matrices \eqref{Eq:DataMatrices} result from the dynamics \eqref{Eq:LinearSubsystemDynamics}, they satisfy:
\begin{equation}\label{Eq:DataDynamics}
\begin{aligned}
\bar{X}_i =&\ A_i X_i + B_i U_i + W_i,\\
Y_i =&\ C_i X_i.
\end{aligned}
\end{equation}
Consequently, the set of model parameters that is consistent under the observed data matrices $\bar{\Phi}_i$ for some disturbance data matrix $W_i$ satisfying \eqref{Eq:As:DisturbanceData}, can be defined as:
\begin{equation}
\begin{aligned}
\Psi_i \triangleq \Big\{(A_i,B_i,C_i):\ &\eqref{Eq:DataDynamics} \mbox{ holds for some } \\
&W_i \mbox{ satisfying } \eqref{Eq:As:DisturbanceData}\Big\}.
\end{aligned}   
\end{equation}

Using \eqref{Eq:DataDynamics}, for any $W_i$ that satisfies \eqref{Eq:As:DisturbanceData}, we can obtain the following quadratic matrix inequality (QMI) constraint on $A_i$ and $B_i$ in $\Theta_i$:
\begin{equation}\label{Eq:QMIonParameters}
\bm{\I \\ A_i^\T \\ B_i^\T}^\T \underbrace{\bm{\I & \bar{X}_i \\ \0 & -X_i \\ \0 & -U_i} Q_{wi} \bm{\I & \bar{X}_i \\ \0 & -X_i \\ \0 & -U_i}^\T}_{\triangleq \bar{Q}_{wi}} \bm{\I \\ A_i^\T \\ B_i^\T} \geq 0.
\end{equation}
On the other hand, using \eqref{Eq:DataDynamics} and As. \ref{As:PersistenceOfExcitation}, we can obtain the following constraint on $C_i$ in $\Theta_i$:  
\begin{equation}\label{Eq:EqualityConstraintOnC}
C_i = Y_i \check{K}_i \mbox{ where } \check{K}_i \in \Pi_i \triangleq\{\check{K}_i : X_i \check{K}_i = \I\}. 
\end{equation}

Using Assumptions \ref{As:DisturbanceData} and \ref{As:PersistenceOfExcitation}, we have the following lemma. 

\begin{lemma}\label{Lm:Equivalence}
Under Assumptions \ref{As:DisturbanceData} and \ref{As:PersistenceOfExcitation}, we have 
$$
\begin{aligned}
\Psi_i = \Gamma_i \triangleq 
\Big\{(A_i,B_i,C_i):\ &\eqref{Eq:QMIonParameters} \mbox{ holds and } \\
&C_i=Y_i\check K_i \mbox{ for some } \check K_i\in\Pi_i
\Big\},  
\end{aligned}
$$
and for any $\Theta_i \in \Psi_i$, \eqref{Eq:EqualityConstraintOnC} uniquely determines $C_i$.
\end{lemma}
\begin{proof}
Since the output data are assumed to be noise-free and generated by the subsystem output equation, there exists a true output matrix $C_i^\star$ such that $Y_i=C_i^\star X_i$. Moreover, by As. \ref{As:PersistenceOfExcitation}, $X_i$ has full row rank, and hence $\Pi_i=\{\check K_i:X_i\check K_i=\I\}$ is nonempty.

First, we prove that $\Gamma_i \subseteq \Psi_i$ by showing that any $\Theta_i \in \Gamma_i \implies \Theta_i \in \Psi_i$. 
Consider a particular $\Theta_i \in \Gamma_i$, which implies that \eqref{Eq:QMIonParameters} and \eqref{Eq:EqualityConstraintOnC} hold for that $\Theta_i$. 
Using this $\Theta_i$ and the observed data matrices $\bar{\Phi}_i$, we can define a disturbance matrix as $W_i(\Theta_i) \triangleq \bar{X}_i - A_i X_i - B_i U_i$. 
Clearly, this $W_i$ satisfies 
$$\bm{\I \\ W_i^\T} = \bm{\I & \bar{X}_i \\ \0 & -X_i \\ \0 & -U_i}^\T \bm{\I \\ A_i^\T \\ B_i^\T},$$ and thus, it also satisfies \eqref{Eq:As:DisturbanceData} (via \eqref{Eq:QMIonParameters}). Moreover, by definition, for this $W_i$, $\Theta_i$ (particularly its $A_i$ and $B_i$) satisfies the first equation in \eqref{Eq:DataDynamics}. 

On the other hand, from \eqref{Eq:EqualityConstraintOnC}, we have $C_i=Y_i\check K_i$ for some $\check K_i\in\Pi_i$. Using $Y_i=C_i^\star X_i$ and $X_i\check K_i=\I$, we obtain
$
C_i=Y_i\check K_i=C_i^\star X_i\check K_i=C_i^\star.
$
Therefore, $Y_i=C_i^\star X_i=C_iX_i$, and hence $C_i$ satisfies the second equation in \eqref{Eq:DataDynamics}. Thus, $\Theta_i$ satisfies both equations in \eqref{Eq:DataDynamics} for a disturbance matrix $W_i$ satisfying \eqref{Eq:As:DisturbanceData}. Therefore, $\Theta_i\in\Psi_i$, and we conclude that $\Gamma_i\subseteq\Psi_i$. 

Next, we prove that $\Gamma_i \supseteq \Psi_i$ by showing that any $\Theta_i \in \Psi_i \implies \Theta_i \in \Gamma_i$. Note that, if any $\Theta_i \in \Psi_i$, that $\Theta_i$ satisfies \eqref{Eq:DataDynamics} for some $W_i$, where $W_i = \bar{X}_i - A_i X_i - B_i U_i$ and satisfies \eqref{Eq:As:DisturbanceData}. Using these two facts about $W_i$, we can obtain \eqref{Eq:QMIonParameters}, implying $\Theta_i$ (particularly its $A_i$ and $B_i$) satisfies \eqref{Eq:QMIonParameters}. 

On the other hand, as $\Theta_i$ satisfies \eqref{Eq:DataDynamics}, we have $Y_i = C_i X_i$. Let us select $\check{K}_i$ such that $X_i \check{K}_i = \I$. Using these two equations, we get $C_iX_i \check{K}_i = C_i \iff Y_i\check{K}_i = C_i$. This implies that $C_i$ satisfies \eqref{Eq:EqualityConstraintOnC}. Therefore, we can conclude that any $\Theta_i \in \Psi_i \implies \Theta_i \in \Gamma_i$, in other words, $\Gamma_i \supseteq \Psi_i$.

Combining these results gives $\Psi_i=\Gamma_i$. Finally, \eqref{Eq:EqualityConstraintOnC} uniquely determines $C_i$. Indeed, for any $\check K_i\in\Pi_i$,
\[
Y_i\check K_i=C_i^\star X_i\check K_i=C_i^\star.
\]
Thus, $Y_i\check K_i$ is independent of the particular right inverse $\check K_i$, and all admissible choices recover the same output matrix $C_i=C_i^\star$. This completes the proof.
\end{proof}

\begin{remark}
Measurement noise in the recorded state and output data can also be incorporated by augmenting the uncertainty description in \eqref{Eq:DataDynamics}. In that case, the boundedness assumption \eqref{Eq:As:DisturbanceData} should be extended to account for measurement noise, and the QMI \eqref{Eq:QMIonParameters} should be replaced by a unified data-consistency condition involving all subsystem parameters $A_i,B_i,C_i$, rather than recovering $C_i$ uniquely through \eqref{Eq:EqualityConstraintOnC}. Similar data-consistency characterizations arise in dualization-based output-feedback data-driven synthesis methods \cite{kristovic2024output}. The results in this paper can be extended along these lines, but we omit this generalization to keep the exposition focused on the core disturbance-robust data-driven design mechanism.
\end{remark}

Now, the goal is to design a local controller $K_i$ \eqref{Eq:SubsystemController} at each subsystem $\Sigma_i, i\in\N_N$, to make the closed-loop subsystem \eqref{Eq:ClosedLoopSubsystemDynamics} $\X_i$-dissipative under all possible model parameters $\Theta_i \in \Psi_i$ given the observed data matrices $\bar{\Phi}_i$. The required data-driven local controller design process is proven in the following proposition.

\begin{proposition}\label{Pr:DataDrivenLocalControllerDesign}
At each subsystem $\Sigma_i, i \in \N_N$, using the observed data matrices $\bar{\Phi}_i$ and under Assumptions \ref{As:DisturbanceData} and \ref{As:PersistenceOfExcitation}, to enforce/optimize $\X_i$-dissipativity for the closed-loop subsystem dynamics \eqref{Eq:ClosedLoopSubsystemDynamics} (with $\X_i^{11}>0$ and $\X_i^{22}<0$) for any subsystem model parameters $\Theta_i \in \Psi_i$, the local controller $K_i$ \eqref{Eq:SubsystemController} can be designed using the LMI problem:
\begin{equation}\label{Eq:Pr:DataDrivenLocalControllerDesign}
\begin{aligned}
\mbox{Find: } &\bar{K}_i, \tilde{K}_i, \X_i, \lambda_i\\
\mbox{Sub. to: } &X_i\tilde K_i>0,\ \ \X_i^{11}>0,\ \ \X_i^{22}<0,\ \ \lambda_i \geq 0,\\
&Q_i > 0,\ \ \mbox{and}\ \ \bm{Q_i & S_i \\ S_i^\T & R_i} \geq 0,
\end{aligned}
\end{equation}
where $Q_i,S_i,R_i$ are as given in \eqref{Eq:Pr:DataDrivenLocalControllerDesignDefinitions}, and $K_i{=}\bar{K}_i(X_i \tilde{K}_i)^{-1}$.
\end{proposition}

\begin{figure*}
\begin{equation}\label{Eq:Pr:DataDrivenLocalControllerDesignDefinitions}
\begin{aligned}
Q_i{\triangleq}\bm{ 
(-\X_i^{22})^{-1} & Y_i \tilde{K}_i & \0 \\
\star & X_i \tilde{K}_i  & \tilde{K}_i^\T Y_i^\T \X^{21}_i\\
\star & \star & \X^{11}_i},\ 
S_i{\triangleq}\bm{\0 & \0 & \0 \\ \0 & X_i \tilde{K}_i & \bar{K}_i^\T \\ \I & \0 & \0},\ 
R_i{\triangleq}\bm{X_i \tilde{K}_i & \0 & \0 \\ \star & \0 & \0 \\ \star & \star & \0 } - \lambda_i \bm{\I & \bar{X}_i \\ \0 & -X_i \\ \0 & -U_i} Q_{wi}\star
\end{aligned}
\end{equation}
\hrulefill
\end{figure*}

\begin{proof}
We start by applying Prop. \ref{Pr:LocalControllerDesign}, which requires:
$$
\bm{ 
(-\X_i^{22})^{-1} & \0 & C_iP_i & \0 \\
\star & P_i & A_iP_i+B_i\bar{K}_i & \I \\
\star & \star & P_i  & P_iC_i^\T \X^{21}_i\\
\star & \star & \star & \X^{11}_i } \geq 0. 
$$
Using Co. \ref{Co:SchurCongruence}, we can obtain an equivalent condition as
$$P_i - \bm{\star}^\T Q_i^{-1} \bm{\0 \\ P_i A_i^\T + \bar{K}_i^\T B_i^\T \\ \I} \geq 0$$
where 
$$Q_i \triangleq  
\bm{ 
(-\X_i^{22})^{-1} & C_iP_i & \0 \\
\star & P_i  & P_iC_i^\T \X^{21}_i\\
\star & \star & \X^{11}_i }. 
$$
This condition can be restated as a QMI in $A_i$ and $B_i$ as 
$$\bm{\I \\ A_i^\T \\ B_i^\T}^\T  \bar{R}_i \bm{\I \\ A_i^\T \\ B_i^\T} \geq 0$$
where 
$$
\bar{R}_i \triangleq \bm{P_i & \0 & \0 \\ \0 & \0 & \0 \\ \0 & \0 & \0 } - \bm{\star}^\T Q_i^{-1} \bm{\0 & \0 & \0 \\ \0 & P_i & \bar{K}_i^\T \\ \I & \0 & \0}.
$$

For this QMI to hold for any $\Theta_i \triangleq (A_i,B_i,C_i) \in \Psi_i$ (see also Lm. \ref{Lm:Equivalence} and \eqref{Eq:QMIonParameters}), both necessary and sufficient conditions can be found using the matrix S-lemma (Lm. \ref{Lm:MatrixSLemma}, under its regularity condition) as
$$
\bar{R}_i - \lambda_i \bar{Q}_{wi} \geq 0,
$$
for some $\lambda_i \geq 0$ where we recall
$$
\bar{Q}_{wi} \triangleq \bm{\I & \bar{X}_i \\ \0 & -X_i \\ \0 & -U_i} Q_{wi} \bm{\I & \bar{X}_i \\ \0 & -X_i \\ \0 & -U_i}^\T.
$$
Now, we apply Lm. \ref{Lm:SchursComplement} to obtain an equivalent condition as
$$
\bar{R}_i - \lambda_i \bar{Q}_{wi} \geq 0 \iff \bm{ Q_i & S_i \\ S_i^\T & R_i } \geq 0,
$$
where $Q_i$ is as defined before,  
\begin{align*}
S_i \triangleq& \bm{\0 & \0 & \0 \\ \0 & P_i & \bar{K}_i^\T \\ \I & \0 & \0},\ \mbox{ and }\ 
R_i \triangleq
\bm{P_i & \0 & \0 \\ \0 & \0 & \0 \\ \0 & \0 & \0 } - \lambda_i \bar{Q}_{wi}. 
\end{align*}

Note that this LMI contains $C_i$ terms, particularly inside its block $Q_i$. As established in Lm. \ref{Lm:Equivalence}, such terms can be replaced using \eqref{Eq:EqualityConstraintOnC}. Consequently, $Q_i$ can be restated as   
$$
Q_i =  
\bm{ 
(-\X_i^{22})^{-1} & Y_i \check{K}_iP_i & \0 \\
\star & P_i  & P_i \check{K}_i^\T Y_i^\T \X^{21}_i\\
\star & \star & \X^{11}_i }, 
$$
where $\check{K}_i$ satisfies $X_i \check{K}_i = \I$. To remove the bilinear terms in $\check{K}_i$ and $P_i$, we next use the change of variables: 
$$
\tilde{K}_i \triangleq \check{K}_i P_i \mbox{ and } P_i = X_i \tilde{K}_i,
$$
where the latter is due to $X_i \check{K}_i = \I \iff X_i \check{K}_i P_i = P_i$. Upon making these substitutions, we obtain the main LMI condition in \eqref{Eq:Pr:DataDrivenLocalControllerDesign} along with the same $Q_i, S_i$ and $R_i$ definitions given in \eqref{Eq:Pr:DataDrivenLocalControllerDesignDefinitions}. 

Finally, as $K_i = \bar{K}_i P_i^{-1}$ according to Prop. \ref{Pr:LocalControllerDesign}, here, it becomes $K_i = \bar{K}_i (X_i \tilde{K}_i)^{-1}$ due to the introduced change of variables. This completes the proof.
\end{proof}

The above result is the data-driven version of Prop. \ref{Pr:LocalControllerDesign}, and thus Rm. \ref{Rm:LocalXiOptimization} also applies here. As the next step, we present the data-driven version of Prop.  \ref{Pr:LocalControllerDesignRevisit}, which is aimed at designing the local controllers so that they positively impact the feasibility and the effectiveness of the subsequent global co-design stage.

\begin{proposition}\label{Pr:DataDrivenLocalControllerDesignRevisit}
At each subsystem $\Sigma_i, i \in \N_N$, using the observed data matrices $\bar{\Phi}_i$ and under Assumptions \ref{As:ScalarDissipativityMatrices}, \ref{As:DisturbanceData} and \ref{As:PersistenceOfExcitation}, to enforce/optimize $\X_i$-dissipativity for the closed-loop subsystem dynamics \eqref{Eq:ClosedLoopSubsystemDynamics} (from $\tilde{u}_i$ to $y_i$, where $\X_i \triangleq [\x_i^{kl}\I]_{k,l\in\N_2}$ with $\x_i^{11} > 0$ and $\x_i^{22} < 0$) for any subsystem model parameters $\Theta_i \in \Psi_i$ so that it favors the subsequent global co-design stage, the local controller $K_i$ \eqref{Eq:SubsystemController} can be designed using the LMI problem:
\begin{equation}\label{Eq:Pr:DataDrivenLocalControllerDesignRevisit}
\begin{aligned}
\mbox{Find: } &\tilde{K}_i, \hat{K}_i, \lambda_{i1}, \lambda_{i2}, \{\bar{\x}_i^{11}, \x_i^{22}\}, \\ 
&\{\tilde{K}_{ii}, \hat{K}_{ii}, \bar{\y}^{11}_{ii}, \bar{\y}_{ii}^{21}, \bar{\y}_{ii}^{22}\},\\
\mbox{Sub. to: } &X_i \hat{K}_i > 0,\ \ \bar{\x}_i^{11}>0,\ \ \x_i^{22} < 0,\ \ \bar{\y}^{11}_{ii} > 0,\\ 
&\bar{\y}^{22}_{ii} < 0,\ \ \lambda_{i1} \geq 0,\ \ \lambda_{i2} \geq 0,\ \ Q_{i1}>0,\\ 
&\bm{Q_{i1} & S_{i1} \\ S_{i1}^\T & R_{i1}} \geq 0,\ \ \mbox{and}\ \ \bm{Q_{i2} & S_{i2} \\ S_{i2}^\T & R_{i2}} \geq 0,
\end{aligned}
\end{equation}
where $Q_{i1},S_{i1},R_{i1}$ and $Q_{i2},S_{i2},R_{i2}$ are as given in \eqref{Eq:Pr:DataDrivenLocalControllerDesignRevisitDefinitions}, and $K_i = \tilde{K}_i (X_i \hat{K}_i)^{-1}$ and $\x_i^{11} = (-\x_i^{22})^{-1} \bar{\x}^{11}_i$.
\end{proposition}

\begin{figure*}
\begin{equation}\label{Eq:Pr:DataDrivenLocalControllerDesignRevisitDefinitions}
\begin{gathered}
Q_{i1} \triangleq \bm{ 
\I & Y_i \hat{K}_i & \0 \\
\star & X_i \hat{K}_i  & \x^{21}_i \hat{K}_i^\T Y_i^\T\\
\star & \star & \bar{\x}^{11}_i\I 
},\
S_{i1} \triangleq \bm{\0 & \0 & \0 \\ \0 & X_i \hat{K}_i & \tilde{K}_i^\T \\ (-\x_i^{22})\I & \0 & \0},\ 
R_{i1} \triangleq \bm{X_i \hat{K}_i & \0 & \0 \\ \star & \0 & \0 \\ \star & \star & \0 } - \lambda_{i1} \bar{Q}_{wi},\\ 
Q_{i2} \triangleq 
\bm{
-\bar{\y}^{22}_{ii}\I & \0 \\ 
\star & \bar{\y}^{11}_{ii}\I 
},\ 
S_{i2} \triangleq 
\bm{
\0 & \0 & \0 & -\bar{\y}^{22}_{ii}\I & \0 & \0\\
\bar{\x}_i^{11}\I & \0 & \0 & -\x_i^{21}\I + \bar{\y}^{21}_{ii}\I & \0 & \0},\ 
R_{i2} \triangleq  
\bm{\bar{R}_{i2}^{11} & \bar{R}_{i2}^{12} \\ \star & \bar{R}_{i2}^{22}} 
- \lambda_{i2}\bm{\bar{Q}_{wi} & \0 \\ \0 & \bar{Q}_{wi}},\\
\bar{R}_{i2}^{11} \triangleq \bm{ \bar{\x}_i^{11}\I & \0 & \0 \\ \star & \0 & \0 \\ \star & \star & \0},\     
\bar{R}_{i2}^{12} \triangleq \bm{ \0 & \0 & \0 \\ \0 & \0 & \0 \\ \tilde{K}_{ii} & \0 & \0},\  
\bar{R}_{i2}^{22} \triangleq \bm{ \I & \0 & -\hat{K}_{ii}^\T \\ \star & \0 & \0 \\ \star & \star & \0} 
\end{gathered}
\end{equation}
\hrulefill
\end{figure*}

\begin{proof}
We start by applying Prop. \ref{Pr:LocalControllerDesignRevisit}, which requires two main LMI conditions, where the first main LMI condition is:
$$
\bm{ 
\I & \0 & C_i\bar{P}_i & \0 \\
\star & \bar{P}_i & A_i\bar{P}_i+B_i\tilde{K}_i & (-\x_i^{22})\I \\
\star & \star & \bar{P}_i  & \x^{21}_i \bar{P}_iC_i^\T \\
\star & \star & \star & \bar{\x}^{11}_i\I 
} \geq 0.
$$
For this, by applying Co. \ref{Co:SchurCongruence}, an equivalent condition can be obtained as
$$
\bar{P}_i - [\star]^\T Q_{i1}^{-1} \bm{\0 \\ \bar{P}_i A_i^\T + \tilde{K}_i^\T B_i^\T \\ (-\x_i^{22})\I} \geq 0,
$$
where
$$
Q_{i1} \triangleq 
\bm{ 
\I & C_i\bar{P}_i & \0 \\
\star & \bar{P}_i  & \x^{21}_i \bar{P}_iC_i^\T \\
\star & \star & \bar{\x}^{11}_i\I 
}.
$$
This LMI condition can now be restated as a QMI in $A_i$ and $B_i$ as
$$
\bm{\I \\ A_i^\T \\ B_i^\T}^\T  \bar{R}_{i1} \bm{\I \\ A_i^\T \\ B_i^\T} \geq 0,
$$
where 
$$
\bar{R}_{i1} \triangleq 
\bm{\bar{P}_i & \0 & \0 \\ \0 & \0 & \0 \\ \0 & \0 & \0 } - \bm{\star}^\T Q_{i1}^{-1} \bm{\0 & \0 & \0 \\ \0 & \bar{P}_i & \tilde{K}_i^\T \\ (-\x_i^{22})\I & \0 & \0}.
$$
For this QMI to hold for any $(A_i,B_i,C_i)\in\Psi_i$, the matrix S-lemma (Lm. \ref{Lm:MatrixSLemma}, under its regularity condition) gives a necessary and sufficient condition as
$\bar{R}_{i1} - \lambda_{i1} \bar{Q}_{wi} \geq 0$, for some $\lambda_{i1} \geq 0$. Applying Lm. \ref{Lm:SchursComplement}, we get an equivalent condition for this as 
$$
\bm{Q_{i1} & S_{i1} \\ S_{i1}^\T & R_{i1}} \geq 0,
$$
where $Q_{i1}$ is as defined before, 
$$
S_{i1} \triangleq \bm{\0 & \0 & \0 \\ \0 & \bar{P}_i & \tilde{K}_i^\T \\ (-\x_i^{22})\I & \0 & \0}, \mbox{ and }
$$
$$
R_{i1} \triangleq \bm{\bar{P}_i & \0 & \0 \\ \0 & \0 & \0 \\ \0 & \0 & \0 } - \lambda_{i1} \bar{Q}_{wi}.
$$
The $C_i$ terms in this LMI (particularly inside its block $Q_{i1}$) now can be equivalently replaced using \eqref{Eq:EqualityConstraintOnC}, which leads to 
$$
Q_{i1} = \bm{ 
\I & Y_i \check{K}_i \bar{P}_i & \0 \\
\star & \bar{P}_i  & \x^{21}_i \bar{P}_i \check{K}_i^\T Y_i^\T\\
\star & \star & \bar{\x}^{11}_i\I 
},
$$
where $\check{K}_i$ satisfies $X_i \check{K}_i = \I$. Now, using the change of variables
$$
\hat{K}_i \triangleq \check{K}_i \bar{P}_i \mbox{ with } \bar{P}_i = X_i \hat{K}_i,
$$
where the latter is due to $X_i \check{K}_i = \I \iff X_i \check{K}_i \bar{P}_i = \bar{P}_i$, we can obtain the first main LMI condition given in \eqref{Eq:Pr:DataDrivenLocalControllerDesignRevisit} along with the same $Q_{i1}, S_{i1}$ and $R_{i1}$ definitions given in \eqref{Eq:Pr:DataDrivenLocalControllerDesignRevisitDefinitions}.  

Next, we consider the second main LMI condition required in Prop. \ref{Pr:LocalControllerDesignRevisit}, which takes the form
$$ 
\bm{
\bar{\x}_i^{11}\I & \0 & B_i \tilde{K}_{ii} & \bar{\x}_i^{11}\I \\
\star & -\bar{\y}^{22}_{ii}\I & -\bar{\y}^{22}_{ii}\I & \0 \\ 
\star& \star & -\mathcal{H}(B_i\hat{K}_{ii})+\I & -\x_i^{21}\I + \bar{\y}^{21}_{ii}\I \\
\star & \star & \star &  \bar{\y}^{11}_{ii}\I 
} \geq 0,
$$
where $\hat{K}_{ii}$ is the relaxed variable introduced in \eqref{Eq:NecessaryConditions}; enforcing 
$\hat{K}_{ii}=\theta_i\tilde K_{ii}$ would recover the unrelaxed necessary condition. 
First, using Lm. \ref{Lm:Congruence}, we obtain an equivalent LMI condition for it as
$$
\bm{
-\bar{\y}^{22}_{ii}\I  & \0 & \0 & -\bar{\y}^{22}_{ii}\I \\ 
\star &  \bar{\y}^{11}_{ii}\I & \bar{\x}_i^{11}\I & -\x_i^{21}\I + \bar{\y}^{21}_{ii}\I \\
\star & \star & \bar{\x}_i^{11}\I & B_i \tilde{K}_{ii} \\
\star & \star & \star& -\mathcal{H}(B_i\hat{K}_{ii})+\I
} \geq 0,
$$
which, through applying Lm. \ref{Lm:SchursComplement}, can be restated as 
\begin{align*}
&\bm{ \bar{\x}_i^{11}\I &  B_i \tilde{K}_{ii} \\ \star & -\mathcal{H}(B_i\hat{K}_{ii})+\I} 
- \bar{S}_{i2}^\T Q_{i2}^{-1} \bar{S}_{i2} \geq 0
\end{align*}
where
$$
Q_{i2} \triangleq 
\bm{
-\bar{\y}^{22}_{ii}\I & \0 \\ 
\star & \bar{\y}^{11}_{ii}\I 
} \mbox{ and }
\bar{S}_{i2} \triangleq \bm{ \0 & -\bar{\y}^{22}_{ii}\I \\ \bar{\x}_i^{11}\I & -\x_i^{21}\I + \bar{\y}^{21}_{ii}\I}.
$$ 
To write this condition as a QMI in $A_i$ and $B_i$, we first define $\bar{\Theta}_i^\T \triangleq \bm{\I & A_i & B_i}$, using which we express its first term as 
$$
\bm{ \bar{\x}_i^{11}\I &  B_i \tilde{K}_{ii} \\ \star & -\mathcal{H}(B_i\hat{K}_{ii})+\I} =
\bm{\star}^\T 
\underbrace{\bm{\bar{R}_{i2}^{11} & \bar{R}_{i2}^{12} \\ \star & \bar{R}_{i2}^{22}}}_{\triangleq \bar{R}_{i2}}
\bm{\bar{\Theta}_i & \0 \\ \0 & \bar{\Theta}_i}
$$
where 
\begin{align*}
\bar{R}_{i2}^{11} \triangleq \bm{ \bar{\x}_i^{11}\I & \0 & \0 \\ \0 & \0 & \0 \\ \0 & \0 & \0}, \quad     
\bar{R}_{i2}^{12} \triangleq \bm{ \0 & \0 & \0 \\ \0 & \0 & \0 \\ \tilde{K}_{ii} & \0 & \0}\\ \span
\bar{R}_{i2}^{22} \triangleq \bm{ \I & \0 & -\hat{K}_{ii}^\T \\ \0 & \0 & \0 \\ -\hat{K}_{ii} & \0 & \0}. 
\end{align*}
Next, we express $\bar{S}_{i2}$ (in its second term) as
$$
\bar{S}_{i2} \triangleq \bm{ \0 & -\bar{\y}^{22}_{ii}\I \\ \bar{\x}_i^{11}\I & -\x_i^{21}\I + \bar{\y}^{21}_{ii}\I} 
= \underbrace{\bm{S_{i2}^{11} & S_{i2}^{12} \\ S_{i2}^{21}  & S_{i2}^{22}}}_{\triangleq S_{i2}}
\bm{\bar{\Theta}_i & \0 \\ \0 & \bar{\Theta}_i}
$$
where
\begin{align*}\span
S_{i2}^{11} \triangleq \bm{ \0 & \0 & \0 }, \quad     
S_{i2}^{12} \triangleq \bm{ -\bar{\y}^{22}_{ii}\I & \0 & \0 },\\  \span
S_{i2}^{21} \triangleq \bm{ \bar{\x}_i^{11}\I & \0 & \0 }, \quad 
S_{i2}^{22} \triangleq \bm{ -\x_i^{21}\I + \bar{\y}^{21}_{ii}\I & \0 & \0}. 
\end{align*}
Using these two representations, we obtain the QMI:
$$
\bm{\bar{\Theta}_i & \0 \\ \0 & \bar{\Theta}_i}^\T 
\underbrace{\big( \bar{R}_{i2} - S_{i2}^\T Q_{i2}^{-1} S_{i2}  \big)}_{\triangleq \tilde{R}_{i2}} \bm{\bar{\Theta}_i & \0 \\ \0 & \bar{\Theta}_i} \geq 0.
$$
On the other hand, from \eqref{Eq:QMIonParameters}, we have the QMI: 
$$
\bm{\bar{\Theta}_i & \0 \\ \0 & \bar{\Theta}_i}^\T 
\underbrace{\bm{\bar{Q}_{wi} & \0 \\ \0 & \bar{Q}_{wi}}}_{\triangleq \tilde{Q}_{wi}}
\bm{\bar{\Theta}_i & \0 \\ \0 & \bar{\Theta}_i} \geq 0.
$$
For this latter QMI to imply the former QMI, the matrix S-lemma (Lm. \ref{Lm:MatrixSLemma}, under its regularity condition) gives an equivalent condition as 
$
\tilde{R}_{i2} - \lambda_{i2} \tilde{Q}_{wi} \geq 0,
$
for some $\lambda_{i2} \geq 0$. Applying Lm. \ref{Lm:SchursComplement}, we can obtain an equivalent condition for this as
$$
\bm{Q_{i2} & S_{i2} \\ S_{i2}^\T & R_{i2}} \geq 0,
$$
where $R_{i2} \triangleq  \bar{R}_{i2} - \lambda_{i2} \tilde{Q}_{wi}$, which is the second main LMI condition given in \eqref{Eq:Pr:DataDrivenLocalControllerDesignRevisit} along with the same $Q_{i2}, S_{i2}$ and $R_{i2}$ definitions given in \eqref{Eq:Pr:DataDrivenLocalControllerDesignRevisitDefinitions}. 

Finally, as $K_i = \tilde{K}_i \bar{P}_i^{-1}$ according to Prop. \ref{Pr:LocalControllerDesignRevisit}, here, it becomes $K_i = \tilde{K}_i (X_i \hat{K}_i)^{-1}$ due to the introduced change of variables. This completes the proof.
\end{proof}

\subsection{Data-Driven Global Co-design}

The global co-design problem formulated in Prop. \ref{Pr:GlobalControllerDesign} primarily uses the already enforced closed-loop subsystem dissipativity properties. However, as evident from the $B\bar{K}$ terms appearing in \eqref{Eq:Pr:GlobalControllerDesign}, the global co-design problem also requires the subsystem input matrices $B \triangleq \diag([B_i]_{i\in\N_N})$. 
Consequently, when the subsystem input matrices are unknown, one cannot directly follow the data-driven local controller designs with the model-based global co-design in Prop. \ref{Pr:GlobalControllerDesign}.
To address this issue, using the observed subsystem data matrices $\{\bar{\Phi}_i: i\in\N_N\}$ subject to As. \ref{As:DisturbanceData}, the following proposition presents a data-driven global co-design problem.

\begin{proposition}\label{Pr:DataDrivenGlobalControllerDesign}
Suppose that the local controllers \eqref{Eq:SubsystemController} have been designed (e.g., via \eqref{Eq:Pr:DataDrivenLocalControllerDesignRevisit}) so that each closed-loop subsystem \eqref{Eq:ClosedLoopSubsystemDynamics} is $\X_i$-dissipative (from $\tilde{u}_i$ to $y_i$, with $\X_i^{11}>0$, i.e., As. \ref{As:PositiveDissipativity}) for any subsystem model parameters $\Theta_i \in \Psi_i, \forall i\in \N_N$ (under Assumptions \ref{As:DisturbanceData} and \ref{As:PersistenceOfExcitation}), and suppose that the desired network-level dissipativity certificate satisfies $\Y^{22}<0$ (i.e., As. \ref{As:NegativeDissipativity}). Then, the distributed global controller and the underlying communication topology  (jointly captured by $K$ in \eqref{Eq:DistributedGlobalController}) can be co-designed to enforce/optimize the $\Y$-dissipativity (from $w$ to $z$) of the closed-loop linear networked system (for any subsystem model parameters $\Theta_i \in \Psi_i, \forall i\in \N_N$), using the LMI problem:     
\begin{equation}\label{Eq:Pr:DataDrivenGlobalControllerDesign}
\begin{aligned}
\min_{\substack{\bar{K}, \Y, \lambda \\ \{p_i\in\R: i\in\N_N\}}}&\ J \triangleq \Vert \bar{K} \Vert_1 - \phi([p_i]_{i\in\N_N}) + \psi(\Y)\\
\mbox{Sub. to:}&\ p_i > 0, \forall i\in\N_N,\ \Y^{11}>0,\ \Y^{22}<0,\\
&\ \lambda \geq 0,\ \mbox{ and } \ \bm{Q & S \\ S^\T & R} \geq 0,
\end{aligned}    
\end{equation}
where $Q,S,R$ are as given in \eqref{Eq:Pr:DataDrivenGlobalControllerDesignDefinitions}, 
$\Y \triangleq [\Y^{kl}]_{k,l\in\N_2}$, 
$\X_p^{kl} \triangleq \diag([p_i\X_i^{kl}]_{i\in\N_N})$,
$\X^{kl} \triangleq \diag([\X^{kl}_i]_{i\in\N_N}), \forall k,l\in\N_2$,
$\bar{\X}^{11} \triangleq (\X^{11})^{-1}$, 
$\bar{\X}_p^{11} \triangleq \bar{\X}^{11} \X_p^{11} \bar{\X}^{11}$, 
and $\bar{K} \triangleq [\bar{K}_{ij}]_{i,j\in\N_N}$ with $K \triangleq [p_i^{-1} \bar{K}_{ij}]_{i,j\in\N_N}$ (functions $J$, $\phi$, and $\psi$ are interpreted as in Rm. \ref{Rm:GlobalObjective}).
\end{proposition}

\begin{figure*}
\begin{equation}\label{Eq:Pr:DataDrivenGlobalControllerDesignDefinitions}
\begin{gathered}
Q \triangleq \bm{\Y^{11} & \0 \\ \star & -\Y^{22}},\ 
S \triangleq \bm{S_{11} & S_{12} & \0 \\ S_{21} & S_{22} & \0},\ 
\begin{aligned}
S_{11} \triangleq& \bm{\X_p^{11}\bar{\X}^{11} & \0 & \0},\ 
&S_{12} \triangleq \bm{-\X_p^{12}+\Y^{12} & \0 & \0},\\
S_{21} \triangleq& \bm{\0 & \0 & \0},\ 
&S_{22} \triangleq \bm{-\Y^{22} & \0 & \0},
\end{aligned}\\
R \triangleq 
\bm{
\tilde{R}_{11} & \tilde{R}_{12} & \0 \\ \star & -\tilde{R}_{22} & -\tilde{R}_{12}^\T \\ \star & \star & \0 
} - \lambda \bm{E \bar{Q}_{w} E^\T & \0 & \0 \\ \0 & \0 & \0 \\ \0 & \0 & E\bar{Q}_{w} E^\T},\ 
\tilde{R}_{11} \triangleq 
\bm{
\bar{\X}_p^{11} & \0 & \0 \\
\0 & \0 & \0 \\
\0 & \0 & \0 
}, \quad
\tilde{R}_{12} \triangleq 
\bm{
\0 & \0 & \0 \\
\0 & \0 & \0 \\
\bar{K} & \0 & \0 
},\\
\tilde{R}_{22} \triangleq 
\bm{
\X_p^{22} & \0 & \0 \\
\0 & \0 & \0 \\
\0 & \0 & \0 
},\ 
E\triangleq\bm{\diag([E_{1i}]_{i\in\N_N})\\\diag([E_{2i}]_{i\in\N_N})\\\diag([E_{3i}]_{i\in\N_N})},\ 
\begin{aligned}
E_{1i}&\triangleq
\bm{\I_{n_{xi}} & \0_{n_{xi}\times n_{xi}} & \0_{n_{xi}\times n_{ui}}},\\
E_{2i}&\triangleq
\bm{\0_{n_{xi}\times n_{xi}} & \I_{n_{xi}} & \0_{n_{xi}\times n_{ui}}},\\
E_{3i}&\triangleq
\bm{\0_{n_{ui}\times n_{xi}} & \0_{n_{ui}\times n_{xi}} & \I_{n_{ui}}},
\end{aligned}\ 
\bar{Q}_w \triangleq \diag([\bar{Q}_{wi}]_{i\in\N_N})
\end{gathered}
\end{equation} 
\hrulefill
\end{figure*}

\begin{proof}
We start by applying Prop. \ref{Pr:GlobalControllerDesign}, which requires:
$$
\bm{
\X_p^{11} & \0 & \X^{11} B \bar{K} & \X_p^{11} \\
\star & -\Y^{22} & -\Y^{22} & \0 \\ 
\star& \star & -\mathcal{H}(\X^{21} B \bar{K})-\X_p^{22} & -\X_p^{21}+\Y^{21} \\
\star & \star & \star &  \Y^{11}
} \geq 0.
$$
Using Lm. \ref{Lm:Congruence}, we get the equivalent LMI condition
$$
\bm{
\bar{\X}_p^{11} & \0 & B \bar{K} & \bar{\X}^{11}\X_p^{11} \\
\star & -\Y^{22} & -\Y^{22} & \0 \\ 
\star& \star & -\mathcal{H}(\X^{21} B \bar{K}) - \X_p^{22} & - \X_p^{21} + \Y^{21} \\
\star & \star & \star &  \Y^{11}
} \geq 0,
$$
where we have used the notations 
$\bar{\X}^{11} \triangleq (\X^{11})^{-1}$, and 
$\bar{\X}_p^{11} \triangleq \bar{\X}^{11} \X_p^{11} \bar{\X}^{11}$ from \eqref{Eq:Pr:DataDrivenGlobalControllerDesignDefinitions}. 
By a block-permutation congruence, we restate the above LMI condition as
$$
\bm{
\Y^{11} & \0 & \X_p^{11}\bar{\X}^{11} & - \X_p^{12} + \Y^{12} \\
\star & -\Y^{22} & \0 & -\Y^{22} \\ 
\star & \star & \bar{\X}_p^{11} & B \bar{K} \\
\star & \star & \star & -\mathcal{H}(\X^{21} B \bar{K}) - \X_p^{22} \\
} \geq 0,
$$
and then via Lm. \ref{Lm:SchursComplement} (under the imposed constraint $Q>0$), we obtain the equivalent condition 
$$
\bar{R} - 
\bar{S}^\T Q^{-1} \bar{S} \geq 0,
$$
where 
$$
\bar{R} \triangleq \bm{
\bar{\X}_p^{11} & B \bar{K} \\
\star & -\mathcal{H}(\X^{21} B \bar{K}) - \X_p^{22}}, 
$$
$$
Q \triangleq \bm{\Y^{11} & \0 \\ \star & -\Y^{22}}, \mbox{ and }
\bar{S} \triangleq \bm{\X_p^{11}\bar{\X}^{11} & -\X_p^{12}+\Y^{12} \\ \0 & -\Y^{22}}.
$$
Defining $\tilde{\Theta}^\T \triangleq \bm{\I & A & B}$ and $\hat{\Theta}^\T \triangleq \X^{21} \tilde{\Theta}^\T$, the $\bar{R}$ and $\bar{S}$ terms in this equivalent condition can be expressed as  
\begin{gather*}
\bar{R} = 
\bm{\tilde{\Theta} & \0 \\ \0 & \tilde{\Theta} \\ \0 & \hat{\Theta}}^\T 
\underbrace{\bm{\tilde{R}_{11} & \tilde{R}_{12} & \0 \\ \tilde{R}_{21} & -\tilde{R}_{22} & -\tilde{R}_{21} \\ \0 & -\tilde{R}_{12} & \0 }}_{\triangleq \tilde{R}} 
\bm{\tilde{\Theta} & \0 \\ \0 & \tilde{\Theta} \\ \0 & \hat{\Theta}},\ \mbox{ and }\\
\bar{S} = \underbrace{\bm{S_{11} & S_{12} & \0 \\ S_{21} & S_{22} & \0}}_{\triangleq S} \bm{\tilde{\Theta} & \0 \\ \0 & \tilde{\Theta} \\ \0 & \hat{\Theta}},    
\end{gather*}
respectively, where all the inner elements of $\tilde{R}$ and $S$ take the forms given in \eqref{Eq:Pr:DataDrivenGlobalControllerDesignDefinitions}. Consequently, the obtained equivalent condition can be restated as a QMI:
\begin{equation}\label{Eq:Pr:DataDrivenGlobalControllerDesignStep1}
\bm{\tilde{\Theta} & \0 \\ \0 & \tilde{\Theta} \\ \0 & \hat{\Theta}}^\T  
\underbrace{\big( \tilde{R}-S^\T Q^{-1}S \big)}_{\triangleq \check{R}} 
\bm{\tilde{\Theta} & \0 \\ \0 & \tilde{\Theta} \\ \0 & \hat{\Theta}} \geq 0.    
\end{equation}

To obtain a similarly structured QMI from \eqref{Eq:QMIonParameters}, let us first denote $\bar{\Theta}_i^\T \triangleq \bm{\I & A_i & B_i},\ \forall i\in\N_N$ and $\bar{\Theta} \triangleq \diag([\bar{\Theta}_i]_{i\in\N_N})$. Using these notations, we can restate the local QMIs \eqref{Eq:QMIonParameters} as: 
$\bar{\Theta}_i^\T \bar{Q}_{wi} \bar{\Theta}_i \geq 0, \forall i\in\N_N$, which can be collectively represented as 
$$\bar{\Theta}^\T \bar{Q}_{w} \bar{\Theta} \geq 0,$$
where $\bar{Q}_{w} \triangleq \diag([\bar{Q}_{wi}]_{i\in\N_N})$ (as defined in \eqref{Eq:Pr:DataDrivenGlobalControllerDesignDefinitions}). 

Next, we identify the permutation matrix $E$ that maps the subsystem-stacked matrix $\bar{\Theta}$ to the type-stacked matrix $\tilde{\Theta}$, i.e., $\tilde{\Theta}=E\bar{\Theta}$. Recall that 
$\bar{\Theta}_i^\T=\bm{\I & A_i & B_i}, 
\bar{\Theta}\triangleq\diag([\bar{\Theta}_i]_{i\in\N_N})$ and $\tilde{\Theta}^\T=\bm{\I & A & B}$, where $A\triangleq\diag([A_i]_{i\in\N_N})$ and $B\triangleq\diag([B_i]_{i\in\N_N})$. Let 
$E \triangleq \bm{E_1^\T & E_2^\T & E_3^\T}^\T$, $E_k \triangleq \diag([E_{ki}]_{i\in\N_N}), \forall k\in\N_3$, where  
\begin{gather*}
E_{1i} \triangleq \bm{\I_{n_{xi}} & \0_{n_{xi}\times n_{xi}} & \0_{n_{xi}\times n_{ui}}},\\ 
E_{2i} \triangleq \bm{\0_{n_{xi}\times n_{xi}} & \I_{n_{xi}} & \0_{n_{xi}\times n_{ui}}},\\
E_{3i} \triangleq \bm{\0_{n_{ui}\times n_{xi}} & \0_{n_{ui}\times n_{xi}} & \I_{n_{ui}}}.
\end{gather*}
Then, 
$E_{1i}\bar{\Theta}_i=\I_{n_{xi}}$, 
$E_{2i}\bar{\Theta}_i=A_i^\T$, and 
$E_{3i}\bar{\Theta}_i=B_i^\T$. 
Therefore,
\[ E\bar{\Theta} = \begin{bmatrix} \operatorname{diag}\!\left([\I_{n_{xi}}]_{i\in\N_N}\right)\\[2mm] \operatorname{diag}\!\left([A_i^\T]_{i\in\N_N}\right)\\[2mm] \operatorname{diag}\!\left([B_i^\T]_{i\in\N_N}\right) \end{bmatrix} = \begin{bmatrix} \I\\ A^\T\\ B^\T \end{bmatrix} = \tilde{\Theta}. \] Since $E$ only permutes the block rows of $\bar{\Theta}$ into the required type-stacked order, it is a square permutation matrix under this stacking, and hence $E^{-1}=E^\T$. Thus, $\bar{\Theta}=E^\T\tilde{\Theta}$.


Consequently, the aggregate QMI obtained from \eqref{Eq:QMIonParameters} can be written as $\tilde{\Theta}^\T E\bar Q_wE^\T\tilde{\Theta}\geq0$. Since 
$\hat{\Theta}^\T=\X^{21}\tilde{\Theta}^\T$, we have $\hat{\Theta}=\tilde{\Theta}\X^{12}$, and therefore
$
\hat{\Theta}^\T E\bar Q_wE^\T\hat{\Theta}
=
(\X^{12})^\T\tilde{\Theta}^\T E\bar Q_wE^\T\tilde{\Theta}\X^{12}\geq0.
$
Now, combining these two QMIs, we get a unified QMI as 
\begin{equation}\label{Eq:Pr:DataDrivenGlobalControllerDesignStep2}    
\bm{\tilde{\Theta} & \0 \\ \0 & \tilde{\Theta} \\ \0 & \hat{\Theta}}^\T 
\underbrace{\bm{E \bar{Q}_{w} E^\T & \0 & \0 \\ \0 & \0 & \0 \\ \0 & \0 & E\bar{Q}_{w} E^\T}}_{\triangleq \tilde{Q}_{w}}
\bm{\tilde{\Theta} & \0 \\ \0 & \tilde{\Theta} \\ \0 & \hat{\Theta}}
\geq 0.
\end{equation}

For the implication $\eqref{Eq:Pr:DataDrivenGlobalControllerDesignStep2} \implies \eqref{Eq:Pr:DataDrivenGlobalControllerDesignStep1}$ to hold, the matrix S-lemma (Lm. \ref{Lm:MatrixSLemma}, under its regularity condition) gives an equivalent condition as 
$
\check{R} - \lambda \tilde{Q}_{w} \geq 0,
$
for some $\lambda \geq 0$. Applying Lm. \ref{Lm:SchursComplement}, we can obtain an equivalent condition for this as
$$
\bm{Q & S \\ S^\T & R} \geq 0,
$$
where $R \triangleq  \tilde{R} - \lambda \tilde{Q}_{w}$, which is the main LMI condition given in \eqref{Eq:Pr:DataDrivenGlobalControllerDesign} along with the same $Q, S$ and $R$ definitions given in \eqref{Eq:Pr:DataDrivenGlobalControllerDesignDefinitions}. 

Finally, note that the identified condition above is still an LMI in design variables $\bar{K}$ and $\{p_i: i\in\N_N\}$. Therefore, according to Prop. \ref{Pr:GlobalControllerDesign}, we recover the distributed controller gain $K$ from $\bar{K}\triangleq [\bar{K}_{ij}]_{i,j\in\N_N}$ and evaluating $K \triangleq [p_i^{-1} \bar{K}_{ij}]_{i,j\in\N_N}$. This completes the proof.
\end{proof}

\section{Application to DC Microgrid Control}\label{Sec:Results}

To demonstrate the effectiveness of the proposed model-based and data-driven co-design techniques, we consider the DC microgrid (DCMG) control problem, where the goal is to design local and distributed global controllers to simultaneously achieve voltage regulation and balanced current sharing with respect to rated voltages and currents.

\subsection{Problem Formulation}

\subsubsection{\textbf{DCMG Model}} In particular, we consider a DCMG with $N$ distributed generators (DGs) $\{\Sigma_i^{DG}:i\in\N_N\}$ and $L$ lines $\{\Sigma_l^{Line}:l\in\N_L\}$. An example DCMG with 6 DGs and 7 lines is shown in Fig. \ref{Fig:DCMG}. The electrical schematic diagrams of a DG $\Sigma_i^{DG}$ and a connected line $\Sigma_l^{Line}$ in a DCMG are as shown in Fig. \ref{Fig:DCMGSchematic}. We use the convention in which electrical loads are modeled as $ZI$-loads (each with a constant impedance ($Z$) and a constant current ($I$) component) connected locally at each DG. The total load current $I_{Li}(t)$ at $\Sigma_i^{DG}, i\in\N_N$ can be expressed as 
$$
I_{Li}(t) \triangleq I_{Li}^{Z}(t) + I_{Li}^{I}(t) \equiv \frac{1}{R_{Li}}V_i(t) + \bar{I}_{Li},
$$
where $R_{Li}$ is the load's constant resistance and $\bar{I}_{Li}$ is the load's constant current. 
On the other hand, each line $\Sigma_l^{Line}\equiv(\Sigma_i^{DG},\Sigma_j^{DG})$ is modeled as a resistive connection $R_{ij}$ between two corresponding DG points of common coupling (PCCs) with voltages $V_i$ and $V_j$. Thus, the line current $I_l$ contribution from $\Sigma_i^{DG}$ to $\Sigma_j^{DG}$ is determined algebraically by Ohm's law as $I_l=(V_i(t)-V_j(t)/R_{ij}$. Consequently, the dynamic subsystems in the co-design formulation are the DG units, while the contributions from loads and physical topology enter the DGs as a fixed disturbance and an uncontrollable input, respectively, as detailed below.

Applying Kirchhoff's current and voltage laws, we can obtain the state space representation for $\Sigma_i^{DG}$ as 
\begin{align}\nonumber
\underbrace{\bm{\dot{V}_i(t) \\ \dot{I}_{ti}(t)}}_{\dot{x}_i(t)} = 
\underbrace{\bm{-\frac{1}{C_iR_{Li}} & \frac{1}{C_i} \\ -\frac{1}{L_i} & -\frac{R_i}{L_i}}}_{\triangleq A_i} 
\underbrace{\bm{V_i(t) \\ I_{ti}(t)}}_{\triangleq x_i(t)} 
+ \underbrace{\bm{-\frac{1}{C_i} \bar{I}_{Li} \\ 0}}_{\triangleq \theta_i} \\
+ \underbrace{\bm{-\frac{1}{C_i} & 0 \\ 0 & \frac{1}{L_i}}}_{\triangleq B_i}
\underbrace{\bm{I_i(t) \\ V_{ti}(t)}}_{\triangleq u_i(t)}
+ \underbrace{\bm{-\frac{1}{C_i} & 0 \\ 0 & \frac{1}{L_i}}}_{\triangleq B_{wi}} w_i(t),
\label{Eq:DGDynamics1}
\end{align}
where $L_i,R_i,C_i$ are the DG circuit elements (see Fig. \ref{Fig:DCMGSchematic}). The DG state $x_i(t)$ is comprised of the components: $V_i(t)$, which denotes the voltage at the coupling point, and $I_{ti}(t)$, which denotes the internal current. The input $u_i(t)$ driving the DG dynamics is comprised of the components: $I_i(t)$ that represents the net current exported from DG into the physical DCMG line network, and $V_{ti}(t)$ that represents the voltage command applied to the voltage source converter (VSC). The disturbance vector $w_i(t)$  affecting the DG dynamics is comprised of the components that represent the disturbances in $I_i(t)$ and $V_{ti}(t)$, respectively. These disturbance components are assumed to be zero-mean and bounded. Using the $A_i,B_i,B_{wi}$ and $\theta_i$ notations introduced in \eqref{Eq:DGDynamics1}, the dynamics of $\Sigma_i^{DG}$ can be expressed concisely as 
\begin{equation}\label{Eq:DGDynamics}
\dot{x}_i(t) = A_i x_i(t) + \theta_i + B_i u_i(t) + B_{wi} w_i(t).
\end{equation}

Although the DCMG dynamics are written in continuous time for physical interpretability, both continuous-time and sampled-data interpretations are useful in this section. The model-based results reported below use the continuous-time versions of the local dissipativity LMIs stated in the remarks following Props. \ref{Pr:LocalControllerDesign} and \ref{Pr:LocalControllerDesignRevisit}. In contrast, the forthcoming data-driven implementation will use a sampled-data realization of the same DCMG error dynamics so that the trajectory data satisfy the discrete-time data equation used in Sec. \ref{Sec:Data-Driven}. To avoid introducing duplicate notation, the same symbols $A_i$ and $B_i$ are used below for the physical continuous-time matrices, while their sampled-data counterparts will be used when applying the data-driven design.

\subsubsection{\textbf{DCMG Input}} Note that the input $u_i(t)$ of $\Sigma_i^{DG}$ contains an uncontrollable component $I_i(t)$ along with a controllable component $V_{ti}(t)$. In particular, under the sign convention shown in Fig. \ref{Fig:DCMGSchematic}, $I_i(t)$ denotes the net current leaving the DG $i$ bus and flowing into the physical line network. Therefore, by Ohm's law over the resistive lines connected to DG $i$, $I_i(t)$ depends on neighboring DG voltages and takes the form
\begin{align}\nonumber
I_i(t) =&\ \sum_{j\in\E_i^P}\frac{1}{R_{ij}}(V_i(t) - V_j(t)) = \sum_{j\in\N_N} \bar{R}_{ij}V_j(t)\\ 
\label{Eq:DGInputPart1}
=&\ \sum_{j\in\N_N} \bar{R}_{ij} D_j^\T x_j(t),
\end{align}
where $\E_i^P$ represents the set of physically connected neighboring DGs to $\Sigma_i^{DG}$, $R_{ij}$ represents the resistance of the line between $\Sigma_i^{DG}$ and $\Sigma_j^{DG}$, and we denote
\begin{equation*}
\begin{aligned}
\bar{R}_{ij} \triangleq&\  
\begin{cases}
-\frac{1}{R_{ij}}, \quad &j \in \E_i^P,\\
\sum_{k \in \E_i^P}\frac{1}{R_{ik}}, \quad &j=i,\\
0, \quad &j \not\in \E_i^P\cup\{i\},\\
\end{cases}\\
D_i \triangleq&\ \bm{1 & 0}^\T, \quad \quad \quad \quad \forall i\in\N_N.
\end{aligned}
\end{equation*}

Equivalently, if $\mathcal{B}_P\in\R^{N\times L}$ denotes any oriented incidence matrix of the physical DCMG line graph and $G_P\triangleq\diag([1/R_l]_{l\in\N_L})$ denotes the diagonal line-conductance matrix, then $\bar R=\mathcal{B}_P G_P \mathcal{B}_P^\T$. The entrywise definition above is used to avoid introducing additional graph notation and to make the sign convention for $I_i(t)$ explicit.

On the other hand, $V_{ti}(t)$ is determined using a hierarchical controller involving: a steady-state controller (denoted $V_{Sti}$), a local feedback controller (denoted $V_{Lti}(t)$), and a global distributed feedback controller (denoted $V_{Gti}(t)$) implemented over the communication topology of the DCMG. In particular, $V_{ti}(t)$ takes the form 
\begin{equation}\label{Eq:DGInputPart2}
\begin{aligned}
V_{ti}(t) 
=&\ V_{Sti} + V_{Lti}(t) + V_{Gti}(t)\\
=&\ V_{Sti} + K_i^{DG}(x_i(t)-x_{Ei}) + \sum_{j\in\N_N} K_{ij}^{DG} (x_j(t)-x_{Ej}),
\end{aligned}
\end{equation}
where the exact forms of $V_{Sti}$ (steady-state input), $x_{Ei}$ (DG state at the equilibrium), and controller gains $K_i^{DG}, [K_{ij}^{DG}]_{j\in\N_N}$ will be designed in the sequel. 

Now, combining \eqref{Eq:DGInputPart1} and \eqref{Eq:DGInputPart2}, we can express the input $u_i(t)$ of $\Sigma_i^{DG}$ as
\begin{equation}\label{Eq:DGInput}
\begin{aligned}
u_i(t) =&\ \bar{D}_i V_{Sti} + \bm{0 \\ K_i^{DG}(x_i(t)-x_{Ei})} \\
&+ \sum_{j\in\N_N} \bm{\bar{R}_{ij}D_j^\T x_j(t) \\ K_{ij}^{DG}(x_j(t)-x_{Ej})},    
\end{aligned}
\end{equation}
where $\bar{D}_i\triangleq \bm{0 & 1}^\T, \forall i\in\N_N$.

\begin{figure}
    \centering
    \includegraphics[width=0.9\linewidth]{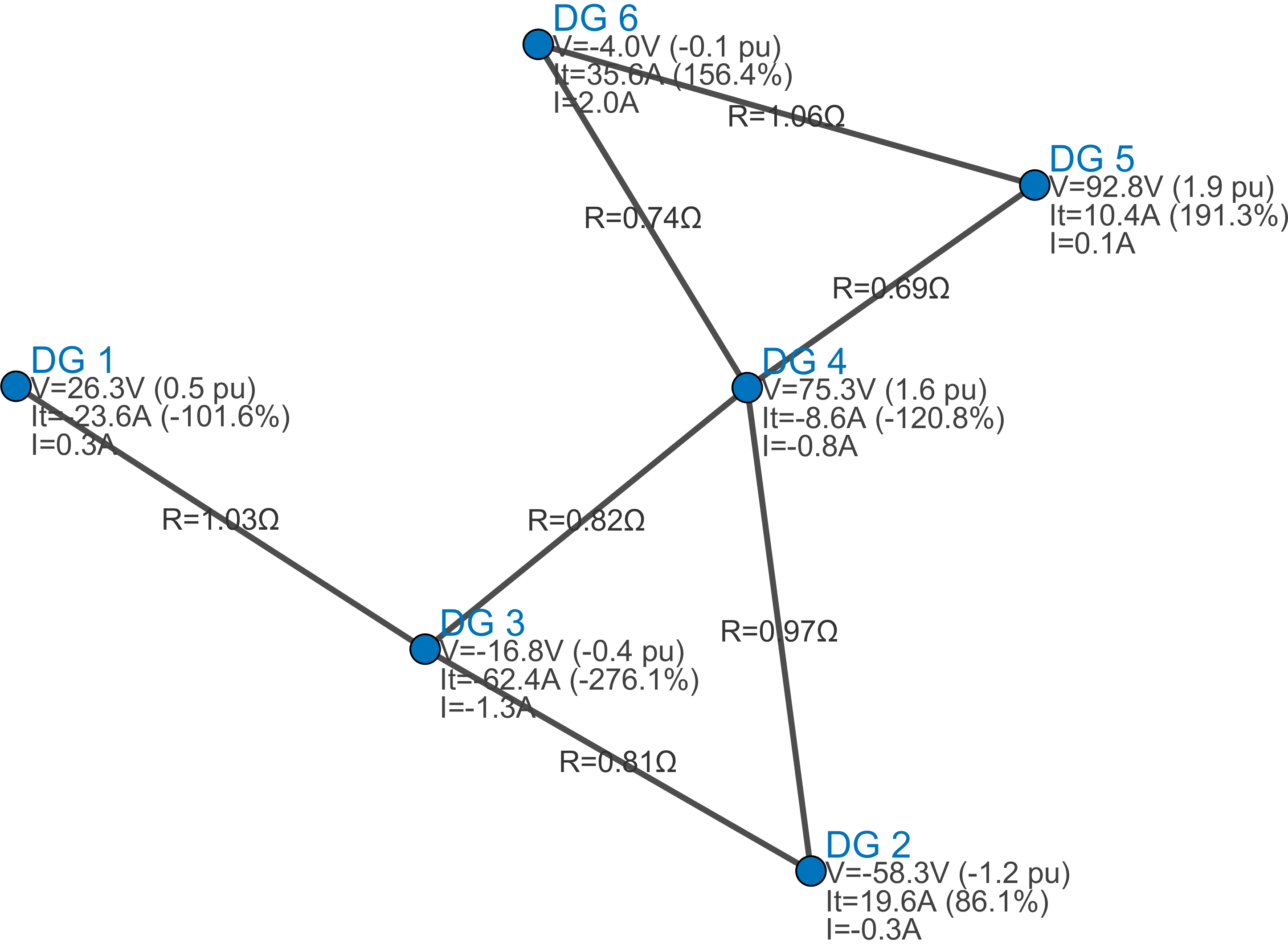}
    \caption{An example DCMG with 6 DGs and 7 lines (initial configuration).}
    \label{Fig:DCMG}
\end{figure}

\begin{figure}
    \centering
    \includegraphics[width=\linewidth]{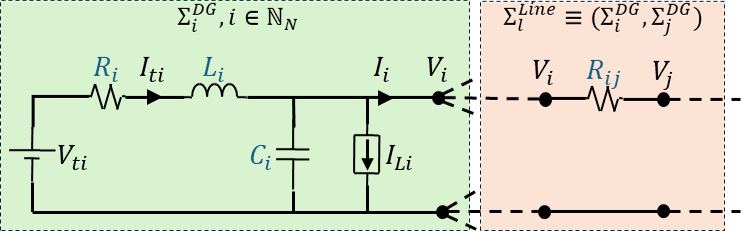}
    \caption{Electrical schematic diagram of a DG $\Sigma_i^{DG}$ and a connected line $\Sigma_l^{Line}$ in a DCMG.}
    \label{Fig:DCMGSchematic}
\end{figure}

\subsubsection{\textbf{Equilibrium Point Design}}
For each DG $\Sigma_i, i\in\N_N$, we denote 
the rated state as $x_{Rat,i} \triangleq \bm{V_{Rat,i} & I_{Rat,ti}}^\T$, the rated input as $u_{Rat,i} \triangleq \bm{I_{Rat,i} & V_{Rat,ti}}^\T$, 
the equilibrium state as $x_{Ei} \triangleq \bm{V_{Ei} & I_{Eti}}^\T$, and 
the equilibrium input as $u_{Ei} \triangleq \bm{I_{Ei} & V_{Eti}}^\T$. 
Note that, while $x_{Rat,i}$ and $u_{Rat,i}$ are known, both $x_{Ei}$ and $u_{Ei}$ are to be designed. In particular, the design specifications for $x_{Ei}$ and $u_{Ei}$ are assumed to be convex constraints of the form:
$$
x_{Ei} \in \mathcal{U}_{xi}(x_{Rat,i}), \quad u_{Ei} \in \mathcal{U}_{ui}(u_{Rat,i}).
$$
For example, to achieve proportional current sharing among the DGs and bound deviations from the rated voltage level (both at the equilibrium), we can constrain $x_{Ei}$ such that:
\begin{equation}\label{Eq:DGSpecs1}
x_{Ei} = \bm{\epsilon_{Vi} & 0 \\ 0 & \epsilon_{Iti}} x_{Rat,i}, \quad 
\vert \epsilon_{Vi} - 1 \vert \leq \delta_{Vi},\quad 
\epsilon_{Iti} = \delta_I,
\end{equation}
for some (or prespecified) $\delta_{Vi} \in [0,1],\ \forall i\in\N_N$ and $\delta_I \in [0,1]$. Similarly, to limit the total current injected to/from DGs and bound the voltage command applied to the VSC (at the equilibrium), we can constrain $u_{Ei}$ so that 
\begin{equation}\label{Eq:DGSpecs2}
u_{Ei} = \bm{\epsilon_{Ii} & 0 \\ 0 & \epsilon_{Vti}} u_{Rat,i}, \quad 
\vert \epsilon_{Ii} \vert \leq \delta_{Ii},\quad 
\vert \epsilon_{Vti} - 1 \vert \leq \delta_{Vti},
\end{equation}
for some (or prespecified) $\delta_{Ii}, \delta_{Vti} \in [0,1],\ \forall i\in\N_N$.

Besides such design specifications, $x_{Ei}$ and $u_{Ei}$ must satisfy the equilibrium conditions (from \eqref{Eq:DGDynamics} and \eqref{Eq:DGInput}):
\begin{equation}\label{Eq:DGSpecs3}
\begin{aligned}
\0 =&\ A_i x_{Ei} + \theta_i + B_i u_{Ei}, \\
u_{Ei} =&\ \bar{D}_i V_{Sti} + \sum_{j\in\N_N} \bm{\bar{R}_{ij}D_j^\T x_{Ej} \\ 0}.
\end{aligned}
\end{equation}
Vectorizing these equilibrium conditions we can obtain:
\begin{equation}\label{Eq:DGSpecs4}
\begin{aligned}
\0 =&\ A x_E + \theta + B u_E, \\
u_E =&\ \bar{D} V_{St} + D \bar{R} D^\T x_E.
\end{aligned}
\end{equation}
where we have defined the vectors:    
$x_E \triangleq [x_{Ei}^\T]_{i\in\N_N}^\T$, 
$u_E \triangleq [u_{Ei}^\T]_{i\in\N_N}^\T$, 
$\theta \triangleq [\theta_i^\T]_{i\in\N_N}^\T$, 
$V_{St} \triangleq [V_{Sti}]_{i\in\N_N}^\T$, 
and the matrices:
$A \triangleq \diag([A_i]_{i\in\N_N})$,
$B \triangleq \diag([B_i]_{i\in\N_N})$,
$D \triangleq \diag([D_i]_{i\in\N_N})$,
$\bar{D} \triangleq \diag([\bar{D}_i]_{i\in\N_N})$,
$\bar{R} \triangleq [\bar{R}_{ij}]_{i,j\in\N_N}$.

In our implementation, the equilibrium point $(x_E, u_E)$ is designed such that constraints \eqref{Eq:DGSpecs1}-\eqref{Eq:DGSpecs4} are satisfied. Note that, through this equilibrium point design, we also indirectly obtain the required steady state control inputs $V_{St}$ (for \eqref{Eq:DGInput}) as the second component of $u_{Ei}$ is $V_{Sti}, \forall i\in\N_N$ (see \eqref{Eq:DGSpecs3}).

\subsubsection{\textbf{Error Dynamics}}
By defining error variables $\tilde{x}_i(t) \triangleq x_i(t) - x_{Ei}$ and $\tilde{u}_i(t) \triangleq u_i(t) - u_{Ei}$ and applying them in \eqref{Eq:DGDynamics} and \eqref{Eq:DGInput}, we obtain the error dynamics of DG $\Sigma_i^{DG}$ as
\begin{equation}\label{Eq:DGErrorDynamics}
\begin{aligned}
\dot{\tilde{x}}_i(t) =&\ A_i \tilde{x}_i(t) + B_i \tilde{u}_i(t) + \tilde{w}_i(t),\\
\tilde{y}_i(t) =&\ C_i \tilde{x}_i(t),\\
\tilde{u}_i(t) =&\ 
K_i \tilde{x}_i(t) + \sum_{j\in\N_N} K_{ij} \tilde{y}_j(t).
\end{aligned}
\end{equation}
where we have used the equilibrium conditions \eqref{Eq:DGSpecs3} and defined 
$\tilde{w}_i(t)\triangleq B_{wi}w_i(t)$, $C_i\triangleq\I$, and the effective local and global gains structured respectively as 
$$
K_i \triangleq \bm{0 \\ K_i^{DG}},
\quad \mbox{ and }\quad 
K_{ij} \triangleq \bm{\bar{R}_{ij}D_j^\T \\ K_{ij}^{DG}}.
$$ 
With these definitions, the DCMG error dynamics have the same structure as the networked-system model used in this paper: the local gain $K_i$ corresponds to the local controller in \eqref{Eq:SubsystemController}, while $K_{ij}$ corresponds to the distributed global controller gain in \eqref{Eq:DistributedGlobalController}.

Clearly, this DG error dynamics, and consequently, the overall closed-loop DCMG error dynamics, are identical to the networked system model used in this paper. Therefore, we can apply both the model-based and data-driven hierarchical control and topology co-design techniques developed in this paper (in Sections \ref{Sec:Model-Based} and \ref{Sec:Data-Driven}, respectively) to design the local feedback controllers, global distributed controllers, and the communication topology for the considered DCMG system, particularly to evaluate the DG control inputs \eqref{Eq:DGInputPart2}.

However, it is worth noting the additional design constraints on the effective local controller gains $K_i, i\in\N_N$ and the effective global distributed controller gains $K_{ij}, i,j\in\N_N$, due to the physical interconnections (lines) in \eqref{Eq:DGErrorDynamics}, that can be expressed as:
\begin{equation}\label{Eq:DGControlGainConstraints}
\begin{aligned}
&D_i^\T K_i = 0,\quad \forall i\in\N_N,\\
&D_i^\T K_{ij} = \bar{R}_{ij}D_j^\T,\quad \forall i,j\in\N_N
\iff D^\T K = \bar{R}D^\T.
\end{aligned}
\end{equation} 
The physical controller constraints in \eqref{Eq:DGControlGainConstraints} can be incorporated directly into the model-based and data-driven LMI formulations given in Props. \ref{Pr:LocalControllerDesign}-\ref{Pr:DataDrivenGlobalControllerDesign} by expressing them in terms of the corresponding decision variables introduced in each proposition. 

For the model-based local controller synthesis problem in Prop. \ref{Pr:LocalControllerDesign}, the controller gain is recovered as $K_i=\bar K_iP_i^{-1}$. Since $P_i$ is nonsingular, the constraint $D_i^\T K_i=0$ is equivalent to $D_i^\T\bar K_i=0$. Therefore, the constraint $D_i^\T\bar K_i=0$ should be added to the LMI problem in Prop. \ref{Pr:LocalControllerDesign}. Similarly, the constraint $D_i^\T\tilde K_i=0$ should be added to the LMI problem in Prop. \ref{Pr:LocalControllerDesignRevisit}. 
For the data-driven local controller synthesis problem of Prop. \ref{Pr:DataDrivenLocalControllerDesign}, the controller gain is recovered as $ K_i=\bar K_i(X_i\tilde K_i)^{-1}$. Therefore, the constraint $D_i^\T K_i=0$ is equivalently enforced by adding $D_i^\T\bar K_i=0$ to the LMI problem in Prop. \ref{Pr:DataDrivenLocalControllerDesign}. Likewise, the constraint $D_i^\T\tilde K_i=0$ should be added to the LMI problem in Prop.  \ref{Pr:DataDrivenLocalControllerDesignRevisit}. 

For the global co-design problems in Props. \ref{Pr:GlobalControllerDesign} and \ref{Pr:DataDrivenGlobalControllerDesign}, the recovered distributed controller gain satisfies $K=[p_i^{-1}\bar K_{ij}]_{i,j\in\N_N}$. The physical constraint $D^\T K=\bar R D^\T$ therefore becomes $D_i^\T(p_i^{-1}\bar K_{ij}) = \bar R_{ij}D_j^\T, \forall i,j\in\N_N,$ or equivalently, \begin{equation}\label{Eq:GlobalConstraintExpanded} 
D_i^\T\bar K_{ij} = p_i\bar R_{ij}D_j^\T, \qquad \forall i,j\in\N_N. 
\end{equation} 
In compact form, defining $P\triangleq\diag([p_i]_{i\in\N_N})$, the constraint \eqref{Eq:GlobalConstraintExpanded} can be written as 
\begin{equation}\label{Eq:GlobalConstraintCompact} 
D^\T\bar K=P\bar R D^\T. 
\end{equation} 
Consequently, the affine equality constraint \eqref{Eq:GlobalConstraintExpanded}, or equivalently \eqref{Eq:GlobalConstraintCompact}, should be added directly to the co-design problems formulated in Props. \ref{Pr:GlobalControllerDesign} and \ref{Pr:DataDrivenGlobalControllerDesign}. Since the constraint is affine in the decision variables $\bar K$ and $p_i$, it preserves the convexity of the corresponding LMI programs.

 
Note also that when applying the data-driven results, the disturbance data matrix $W_i$ in \eqref{Eq:DataMatrices} corresponds to samples of $\tilde{w}_i$, not the primitive converter/input disturbance $w_i$. Thus, the quadratic bound in As. \ref{As:DisturbanceData} should be imposed on the realized additive state-channel disturbance in the error dynamics. If a bound is instead specified on the primitive disturbance $w_i$, it should be mapped through the disturbance channel $B_{wi}$, or through its sampled-data counterpart in the discrete-time implementation, to obtain a valid bound on $W_i$.

\subsection{Results and Discussion}
We first evaluate the model-based hierarchical co-design framework using the DCMG model described above (the data-driven implementation results, along with a more comprehensive comparison, will be added soon). The equilibrium point is selected to satisfy the voltage-regulation and proportional-current-sharing specifications in \eqref{Eq:DGSpecs1}--\eqref{Eq:DGSpecs4}. Around this equilibrium, the local controller design is used to enforce subsystem-level dissipativity certificates, and the global co-design problem is then solved to synthesize the distributed coordination gains and the communication topology. For comparison, we also implement a stabilizing global controller that achieves closed-loop stability but does not explicitly optimize the dissipativity certificate or communication sparsity.

Figure \ref{Fig:DCMG-Comparison} compares the resulting communication topologies and DG state trajectories. The stabilizing design uses a fully connected communication topology, as shown in Fig. \ref{Fig:Stabilizing_DCMG_Topology}, whereas the proposed dissipativity-based hierarchical design yields a sparser communication topology, as shown in Fig. \ref{Fig:DCMG_Topology}. The corresponding trajectories in Figs. \ref{Fig:Stabilizing-Design} and \ref{Fig:Model-Based-Design} show that both controllers regulate the DG voltages and currents to their desired steady-state values, while the proposed hierarchical co-design achieves this with fewer communication links and explicitly enforces the desired dissipativity certificate. 


\begin{figure*}[ht]
    \centering

    \begin{subfigure}{0.48\textwidth}
        \centering
        \includegraphics[width=\linewidth]{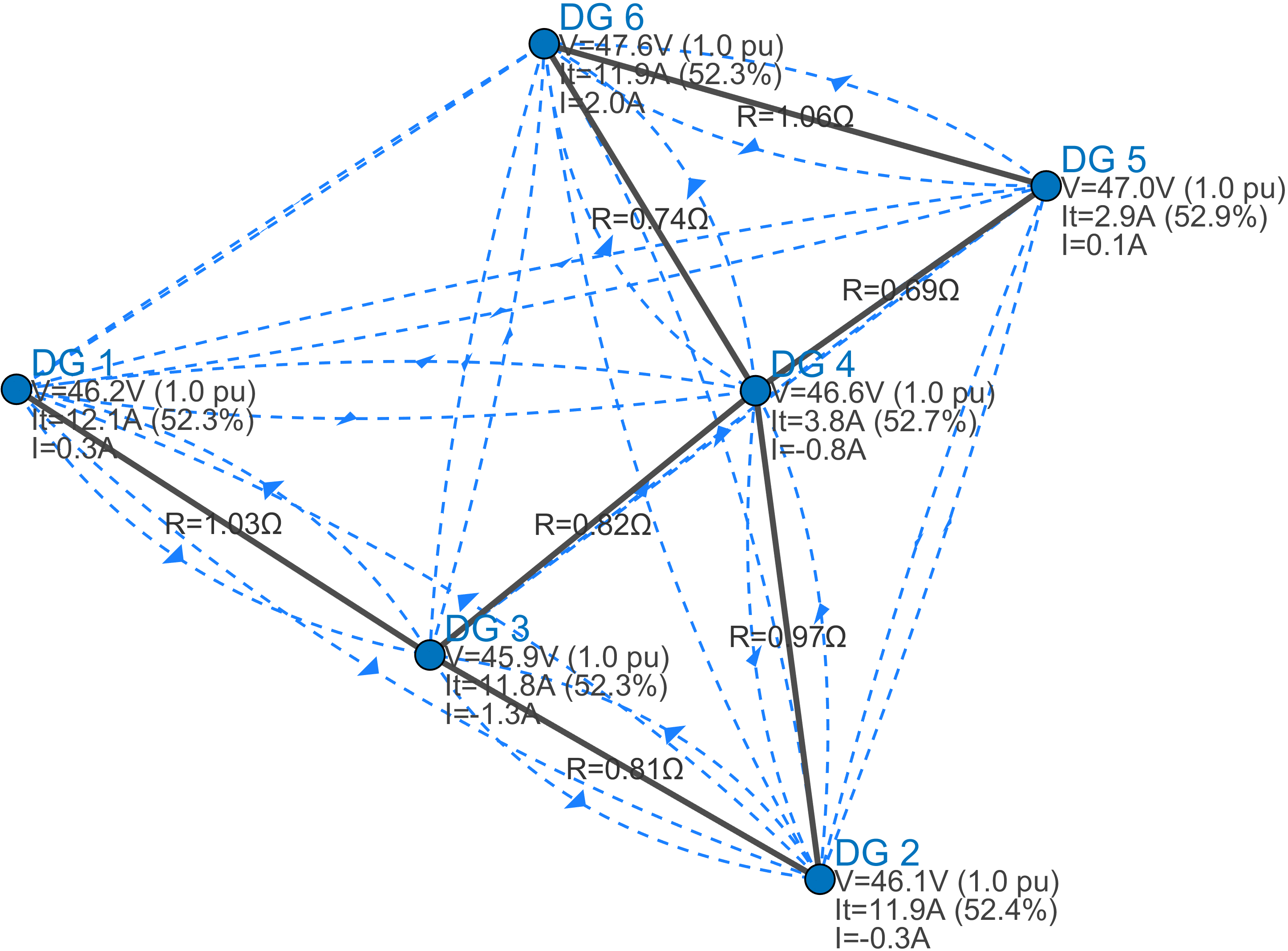}
        \caption{Communication topology (fully-connected) and the final DCMG configuration observed under a model-based stabilizing global DCMG controller.}
        \label{Fig:Stabilizing_DCMG_Topology}
    \end{subfigure}
    \hfill
    \begin{subfigure}{0.48\textwidth}
        \centering
        \includegraphics[width=\linewidth]{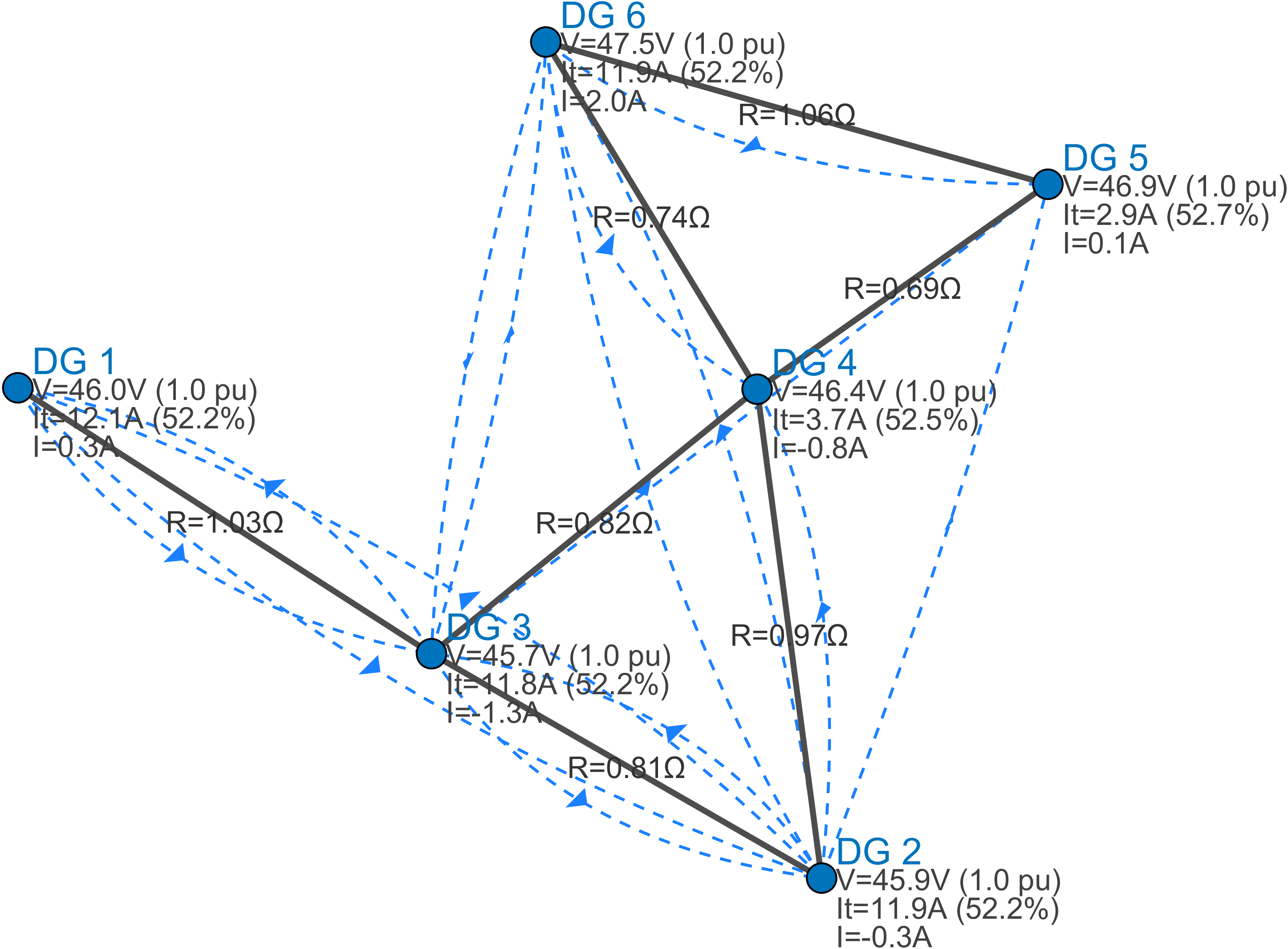}
        \caption{Communication topology and the final DCMG configuration observed under the proposed model-based dissipativating hierarchical distributed DCMG controller.}
        \label{Fig:DCMG_Topology}
    \end{subfigure}
    \begin{subfigure}{0.48\textwidth}
        \centering
        \includegraphics[width=\linewidth]{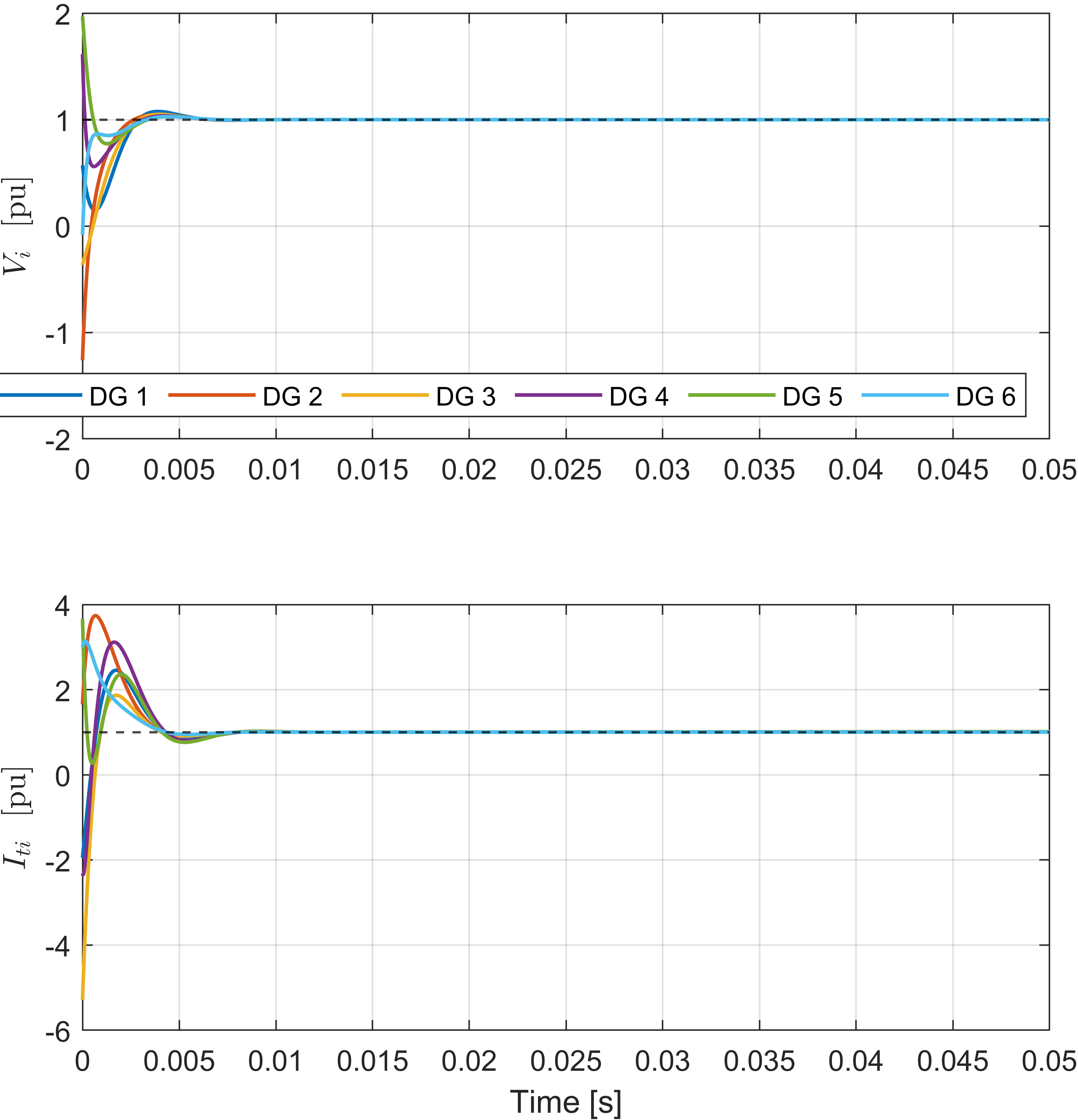}
        \caption{DG state trajectories observed under a model-based stabilizing global DCMG controller.}
        \label{Fig:Stabilizing-Design}
    \end{subfigure}
    \hfill
    \begin{subfigure}{0.48\textwidth}
        \centering
        \includegraphics[width=\linewidth]{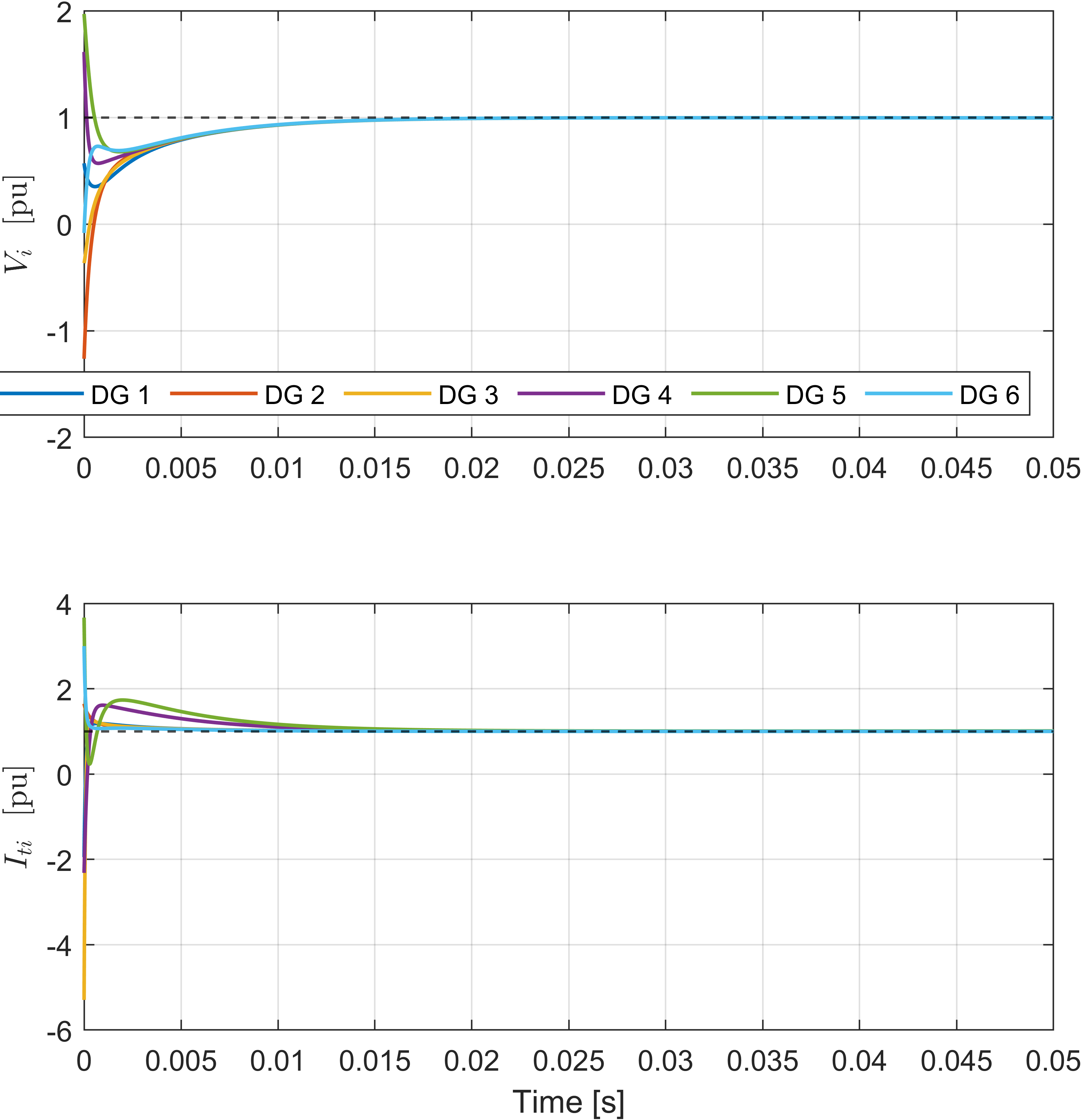}
        \caption{DG state trajectories observed under the proposed model-based dissipativating hierarchical distributed DCMG controller.}
        \label{Fig:Model-Based-Design}
    \end{subfigure}
    \caption{Numerical results comparison between stabilizing global control vs. dissipativating hierarchical distributed control for the DCMG voltage regulation and balanced current sharing tasks.}
    \label{Fig:DCMG-Comparison}
\end{figure*}


\bibliographystyle{IEEEtran}
\bibliography{References}

\end{document}